\documentclass[a4paper,numberwithinsect,cleveref,autoref,thm-restate]{lipics-v2021}
\usepackage{amsthm,amsmath,amsfonts,amssymb,color}
\usepackage{stmaryrd}%
\usepackage{tikz-cd}
\usepackage{tikz}
\usetikzlibrary{arrows,positioning,shapes,decorations,automata,backgrounds,petri,fit,calc}
\usepackage{mathtools}
\usepackage{multicol}
\usepackage{hyperref,cleveref}

\usepackage{todonotes}

\newcommand{\sema}[1]{{\llbracket}#1{\rrbracket}}
\newcommand{\class}[1]{\mathbb{#1}}
\newcommand{\alg}[1]{\textbf{\textsl{#1}}}

\newcommand{\LTL}{\textsf{LTL}}

\newcommand{\LTLfin}{\textsf{LTL}^{\text{\upshape fin}}}

\newcommand{\pLTL}{\ensuremath{\textrm{pLTL}}}

\newcommand{\temporalopsfont}[1]{\texttt{#1}}
\newcommand{\ap}{\temporalopsfont{AP}}
\newcommand{\ltlX}{\temporalopsfont{X}}
\newcommand{\ltlG}{\temporalopsfont{G}}
\newcommand{\ltlGS}{\hat{\temporalopsfont{G}}}
\newcommand{\ltlF}{\temporalopsfont{F}}
\newcommand{\ltlFS}{\hat{\temporalopsfont{F}}}
\newcommand{\ltlU}{\temporalopsfont{U}}
\newcommand{\true}{\temporalopsfont{true}}

\newcommand{\set}[1]{\{#1\}}

\newcommand{\Excl}{\mathrm{Excl}}
\newcommand{\ExclWord}{\mathrm{ExclWord}}

\newcommand{\Ecan}{\mathrm{E_{can}}}
\newcommand{\Dual}{\mathrm{Dual}}
\newcommand{\mergeop}{\mathbin{\bowtie}}

\newcommand{\numOfFS}{\#_{\ltlFS}}
\newcommand{\numOfleaf}{\#_{\text{leaf}}}

\newcommand{\semafin}[1]{\sema{#1}^{\textup{fin}}}

\newcommand{\dom}{\text{dom}}

\newcommand{\words}{\mathcal{W}}
\newcommand{\finwords}{\mathcal{W}^{\textsf{fin}}}
\newcommand{\allwords}{\mathcal{W}^{\textsf{all}}}

\renewcommand{\phi}{\varphi}

\title{Characterizing LTL Formulas by Examples \\ (full version)}



\relatedversion{A full version with detailed proofs is available at~\cite{CateFOS26arxiv}.}

\author{Balder ten Cate}{ILLC, University of Amsterdam, The Netherlands}{b.d.tencate@uva.nl}{https://orcid.org/0000-0002-2538-5846}{}

\author{Dana Fisman}{Institute for the Theory of Computing,  Stein Faculty of Computer and Information Science, \\ Ben Gurion University, Israel}{dana@bgu.ac.il}{https://orcid.org/0000-0002-6015-4170}{}

\author{Roi Ohayon}{Institute for the Theory of Computing, Stein Faculty of Computer and Information Science, \\ Ben Gurion University, Israel}{roioha@post.bgu.ac.il}{https://orcid.org/0009-0009-2439-4466}{Supported by ISF grant 250721}

\author{Patrik Sestic}{ }{patrik.sestic@posteo.com}{https://orcid.org/0009-0001-1212-6896}{}

\authorrunning{B.~ten Cate and D.~Fisman and R.~Ohayon and P. Sestic} 
\Copyright{B.~ten Cate and D. Fisman and R.~Ohayon and P. Sestic} 

\acknowledgements{We thank Frank Wolter and Raoul Koudijs for their input on early versions.}

\EventEditors{Michal Kouck\'{y} and Daniela Petri\c{s}an}
\EventNoEds{2}
\EventLongTitle{51st International Symposium on Mathematical Foundations of Computer Science (MFCS 2026)}
\EventShortTitle{MFCS 2026}
\EventAcronym{MFCS}
\EventYear{2026}
\EventDate{August 24--28, 2026}
\EventLocation{Paris, France}
\EventLogo{}
\SeriesVolume{386}
\ArticleNo{13}

\begin{document}
\nolinenumbers
\maketitle

\begin{abstract}
We investigate the extent to which Linear Temporal Logic (LTL) formulas can be uniquely characterized by a finite set of labeled examples. 
We consider different types of examples, ranging from
finite words to transfinite words, as well as schematic examples.
In the finite-word setting, we provide a complete classification of basis-restricted LTL fragments that admit such unique characterizations. Next, we show that allowing transfinite words as examples enables finite unique characterizations for large monotone fragments of LTL. Finally, we introduce schematic examples, i.e., patterns that compactly represent a family of finite words, and we show that these enable unique characterization results in the finite setting that were not possible with ordinary finite examples alone. Overall, the work provides a foundational account of the descriptive power of different example types for example-driven specification, debugging, and learning of temporal properties.
\end{abstract}

\section{Introduction}

Linear Temporal Logic (LTL) \cite{Pnueli77} is widely used in formal verification, as a language for specifying requirements on the intended behavior of 
systems. In practice, writing LTL
formulas is an error-prone process. Tools often rely on examples, i.e., concrete words (traces) that
either satisfy or violate the specification, because they are easy to inspect and communicate.
Examples arise naturally in automata-theoretic model checking:
when a finite-state system violates an LTL property, one can extract a finite or an ultimately periodic
counterexample~\cite{Vardi95}. Considerable work studies how to produce
useful counterexamples, including short counterexamples \cite{SchuppanBiere05}.
Examples are also used as an interface for understanding and debugging temporal
specifications \cite{AmmonsEtAl03}. Complementarily, there is a growing literature on
learning temporal specifications from labeled examples~\cite{NeiderGavran18,CamachoMcIlraith19},
and on interactive workflows that refine a specification by querying the user with candidate
behaviors \cite{GavranDarulovaMajumdar20}.

In this paper, we study a foundational question that underlies the above work: how much information about an LTL-formula can be conveyed by a finite set of labeled examples? More specifically, 
\emph{when does there exist a finite set of labeled examples that
uniquely identifies a formula (up to semantic equivalence)?}
We  adopt a learner-independent notion of identifiability: a set of labeled examples
\emph{uniquely characterizes} a target formula within a formula class if no
other formula in the class fits it. This is the classical notion of a
\emph{teaching set} from computational learning theory~\cite{GoldmanKearns95}.
We study the existence of such uniquely characterizing sets of examples, 
for different types of examples, including \emph{finite words}, \emph{$\omega$-words}, \emph{transfinite words}, and \emph{schematic examples}.
These different types of examples turn out to differ in their descriptive power when it comes to distinguishing LTL formulas.

\begin{example}
\label{ex:intro}
The LTL-formula $\ltlF p$ 
is distinguished from $\bigvee_{i\leq n} \ltlX^i p$ by
the example $\emptyset^{n+2} \{p\}$. The same formula
cannot, however, be distinguished from all 
formulas of the form 
$\bigvee_{i\leq n} \ltlX^i p$
using only finitely many labeled examples, if the examples are finite
words, or even $\omega$-words. On the other hand,
the transfinite word $\emptyset^\omega \{p\}$, of length $\omega+1$, satisfies 
$\ltlF p$ while falsifying
$\bigvee_{i\leq n} \ltlX^i p$
for all $n$. Similarly, the 
schematic example
$\emptyset^*\{p\}$ is a positive example for
$\ltlF p$, in the sense that all instances of the schema are words satisfying $\ltlF p$, and it is not a positive example in the same sense for 
$\bigvee_{i\leq n} \ltlX^i p$
for any $n$.
\end{example}

We start from the finite-words setting. Building on prior work \cite{FortinEtAlKR2022UniqueTemporalInstanceQueries}, we show that, in this setting, unique characterizability is the exception rather than the rule: for most fragments of LTL, over finite words, no finite set of positive/negative examples can uniquely characterize a formula up to equivalence. Moving from finite words to $\omega$-words improves the picture, but the improvement is limited. A key shift comes when we go beyond $\omega$-words and consider transfinite words as examples. Transfinite ordinal-indexed words allow us to separate formulas in a more robust way, yielding stronger positive results and, in particular, existence of finite uniquely characterizing sets for meaningful fragments of LTL. Finally, we use these transfinite examples as a stepping stone towards new results in  the finite setting: we introduce \emph{schematic examples}, i.e.,  finite patterns that stand for a family of concrete examples. We show how the additional  power available in the transfinite setting can be compiled into such patterns. This yields unique characterization results that are unattainable with concrete finite examples alone, while still keeping the examples close to the finite descriptions %
familiar to practitioners.

More concretely, our main contributions are:
\begin{enumerate}
    \item In the finite-word setting, we provide a complete classification of 
     fragments $\LTL_O$ with 
    $O\subseteq \{\ltlU,\ltlF,\ltlFS,\ltlX,\wedge,\vee,\neg,\top,\bot\}$, with respect to whether the fragment in question admits finite unique characterizations when interpreted over finite words. There are six maximal such fragments that do admit finite unique characterizations, but they are rather limited in expressive power. In particular, as was observed already in~\cite{FortinEtAlKR2022UniqueTemporalInstanceQueries}, even $\LTL_{\ltlF,\land}$ already does not admit finite characterizations, and neither does $\LTL_{\ltlX,\lor}$.
    \item Moving on to $\omega$-words, we show that, although $\LTL_{\ltlX,\land,\lor,\top,\bot}$ admits finite characterizations consisting of $\omega$-words, other fragments such as 
    $\LTL_{\ltlF,\land}$ and $\LTL_{\ltlG,\lor}$ still do not.
    \item We show that $\LTL_{\ltlFS,\land,\lor,\top,\bot}$ admits finite characterizations consisting of (finitely present\-able) transfinite words. A similar result holds for the dual fragment $\LTL_{\ltlGS,\land,\lor,\top,\bot}$.
    \item Returning to the finite-word setting, we introduce the concept of \emph{schematic examples} and we show that $\LTL_{\ltlFS,\land,\lor,\top,\bot}$ admits finite characterizations by schematic examples.
\end{enumerate}

Due to space constraints, most proofs are deferred to the full version~\cite{CateFOS26arxiv}.
A python library and an interactive demonstration of the concepts and algorithms developed in this paper can be found at 
\url{https://ltl-workbench.onrender.com/}.

\paragraph*{Related work}
The problem of inferring temporal specifications from examples has attracted
considerable attention in recent years. SAT-based and decision-tree-based
approaches have been proposed for learning LTL formulas from positive and
negative examples, aiming to find a small formula consistent with the given
sample~\cite{NeiderGavran18}. This task is often described as separating
positive from negative examples; in our terminology, this means finding a
formula that \emph{fits} the sample. Learning LTL over finite traces (LTLf) is
considered in~\cite{RahaRFN26}. An approach combining passive learning with
syntactic constraints on the space of inferred formulas is pursued
in~\cite{Changjian25}. A different direction studies whether neural networks
can learn semantic aspects of temporal logic, such as producing satisfying
traces for LTL formulas~\cite{HahnEtAl21}.

Our work is related to passive learning from examples, but asks a
different question. Passive learning asks whether, given a sample, there exists
a formula in a given class that fits it, often seeking a small %
such 
formula. In contrast, we study when a given LTL formula can be
\emph{uniquely characterized} by a finite set of labeled examples, relative to a class of formulas. Thus, our
focus is not on synthesizing one consistent explanation, but on understanding
which fragments admit finite samples that rule out all
inequivalent alternatives.

A more foundational line of work studies the computational complexity of
passive learning. It is shown that passive learning of LTL formulas is
NP-complete for almost all standard fragments~\cite{MascleFL23}, and this is
extended to LTL, CTL, and ATL, where the problem remains NP-complete with
unbounded binary operators, while bounded-binary fragments exhibit more
fine-grained differences~\cite{BordaisN025}. These works identify fundamental
algorithmic limits of example-based specification: given a sample, decide
whether some formula in the class fits it. In contrast, our focus is semantic:
given a fragment, we determine whether every formula admits a finite sample
that uniquely identifies it up to equivalence.

A related perspective is that of machine teaching. Adaptive teaching of
parametric LTL (\pLTL) formulas to learners with preferences is studied
in~\cite{XuChenTopcu20}.\footnote{In \pLTL, the $\ltlF$ and $\ltlG$ operators are associated with a parameter bounding their number of time steps.}
This line of work is algorithmic and learner-dependent: examples are chosen with
respect to a given learner model and preference structure. In contrast, our
notion of characterizability is learner-independent: a characterizing set must
identify the target formula semantically among all formulas in the class,
independently of any learning algorithm or preference relation. Thus, our
results isolate structural limits of example-based teaching for LTL.

Another related thread aims to explain LTL formulas using words or abstractions
thereof. For example,~\cite{LiEtAlVISSOFT23LTLTimelines} presents visualizations
of an LTL formula’s behavior in a form close to the ``timeline sketches'' used
by practitioners; see also~\cite{WangGamboaRozierSCP2026WEST}. While such work
focuses on making a given formula more understandable, our goal is to understand
when examples themselves determine the formula.

In the knowledge representation literature, motivated by
exact learnability questions in active-learning settings with membership queries, the question has been studied whether a description logic concept is uniquely characterized by a finite set of positive and negative examples \cite{BalderCQ,DBLP:conf/ijcai/FunkJL22}. This study was further extended to temporalized description logic concepts in~\cite{FortinEtAlKR2022UniqueTemporalInstanceQueries,JungEtAlKR2024UniqueTemporalQueriesOntology}. In this context,~\cite{FortinEtAlKR2022UniqueTemporalInstanceQueries} also studied unique characterizations for fragments of LTL over finite words. We will build on some of the results from~\cite{FortinEtAlKR2022UniqueTemporalInstanceQueries}.

LTL over transfinite words was proposed as a means to represent systems with Zeno-like behaviors and multi-phase executions in~\cite{DemriNowak2006,DemriRabinovich2010}, where it was shown  that the usual automata-theoretic techniques for LTL can be extended suitably to transfinite words, yielding satisfiability and model-checking procedures with similar complexity as over $\omega$-words.

\section{Preliminaries}\label{sec:prelim}

\paragraph*{Transfinite words}

We briefly recall standard notions on ordinals needed in the sequel. 
An ordinal is a well-ordered set, where every non-empty subset admits a least element, and ordinals are identified up to order isomorphism. 
Every ordinal is either $0$, a successor ordinal of the form $\beta+1$, or a limit ordinal. 
Ordinal arithmetic (addition, multiplication, and exponentiation) is defined inductively and is, in general, non-commutative. 
In particular, for any ordinals $\alpha < \beta$ there exists a unique ordinal $\gamma$ such that $\beta = \alpha + \gamma$, and for any $\alpha$ and $\beta > 0$ there exist unique ordinals $\gamma$ and $\delta < \beta$ such that $\alpha = \beta \cdot \gamma + \delta$. 
We refer to standard texts (e.g., Jech (1997)) for further details. We will only be concerned with 
\emph{countable} ordinals, i.e., ordinals $\alpha$ such that $\{\beta\mid \beta<\alpha\}$ is a countable set.

\begin{definition}[$\alpha$-words]\label{def:alpha_words}
Fix a finite set of atomic propositions $\ap$.
    If $\alpha$ is a countable ordinal, then an \emph{$\alpha$-word} is a function 
$w : \alpha \to \mathcal{P}(\ap)$.
\end{definition}

We think of an $\alpha$-word $w$ as a (possibly transfinite) sequence of
valuations $w[0]\,w[1]\,w[2]\,\ldots$, indexed by the ordinals below
$\alpha$, where $w[\beta]$ is the valuation at position $\beta$. If
$\alpha>\omega$, the sequence extends beyond the finite positions.
We refer to $\alpha$-words as \emph{ordinal words}. If $\alpha\ge\omega$, we
also refer to them as \emph{transfinite words}.

We also 
write $|w|$ to denote the ordinal $\alpha$.
For $\beta<|w|$, we write $w[\beta..]$ to denote
the suffix of $w$ starting at position $\beta$. Formally,  $w[\beta..]$ is the $(|w|-\beta)$-word $w'$ such that
$w'[\gamma]=w[\beta+\gamma]$.

For a countable ordinal $\alpha$ we
denote by $\words^{\leq \alpha}$ the set of all $\beta$-words $w$ with $\beta\leq \alpha$. Similarly,
we denote by $\words^{<\alpha}$ and $\words^{=\alpha}$
the set of all ordinal words of length less than $\alpha$, respectively, of length exactly $\alpha$.
We will also write $\finwords$ for $\words^{<\omega}$ and $\allwords$ for the set of all countable ordinal words.

Although $\alpha$-words do not necessarily follow a repeating pattern, we will be particularly interested in $\alpha$-words that do, since, despite their possibly infinite length, they can be represented by finite expressions. We introduce the following word-expression language for their representation.

\begin{definition}[Word-expressions]\label{def:word_expr}
A \emph{word-expression} over $\ap$ is any expression generated inductively by the following rules:
\begin{itemize}
    \item Each $\sigma \subseteq \ap$ constitutes  a $1$-word-expression. It denotes the $1$-word $w_\sigma$ where $w_\sigma[0]=\sigma$.
    \item If $A$ is an $\alpha$-word-expression and $B$ is a $\beta$-word-expression, then 
    $A \cdot B$
    is an $(\alpha + \beta)$-word-expression. It denotes the $(\alpha + \beta)$-word
    $w_{A \cdot B}$
    defined by
    \[
    w_{A \cdot B}[\gamma] =
    \begin{cases}
        w_A[\gamma] & \text{if $\gamma \in \alpha$}\\[4pt]
        w_B[\gamma - \alpha], & \text{otherwise}.
    \end{cases}
    \]
    We sometimes omit the concatenation symbol $\cdot$ when no ambiguity arises.

    \item If $A$ is an $\alpha$-word-expression, then 
    $(A)^\omega$
    is an $(\alpha\cdot\omega)$-word-expression. It denotes the 
    $(\alpha\cdot\omega)$-word
    $w_{(A)^\omega}$
    defined, for all $\gamma < \alpha \cdot \omega$, by
    $w_{(A)^\omega}[\gamma] = w_A[\gamma \bmod \alpha]$.
\end{itemize}
\end{definition}

We say that an $\alpha$-word is \emph{finitely presentable} (or, \emph{regular}) if 
there is an $\alpha$-word expression that denotes it.

\begin{example} 
\label{ex:ordinal-word}
    The word-expression
    $\{q\}^\omega\cdot \{p\} \cdot \{q\}^\omega$ denotes a transfinite word of length $\omega+\omega$ that can be depicted as follows: $\{q\}\longrightarrow\{q\}\longrightarrow\ldots ~~ \{p\}\longrightarrow\{q\}\longrightarrow\{q\}\longrightarrow\ldots$.
\end{example}
We note that, although the definitions above apply to arbitrary countable
ordinals, all transfinite words considered in this paper have length below
$\omega^2$. In particular, they contain only finitely many infinite stretches.

\paragraph*{Linear Temporal Logic}

Let $\ap$ be a set of atomic propositions. 
An \emph{$\LTL[\ap]$ formula} is any expression generated
by the following grammar:
\[
\varphi ::= \top 
  \mid p 
  \mid \neg \varphi
  \mid \varphi_1 \land \varphi_2 
  \mid \ltlX \varphi
  \mid \ltlF \varphi
  \mid \varphi_1\, \ltlU\, \varphi_2
\]

We can define some other logical connectives and temporal operators as abbreviations:
\[\begin{array}{lll}
\bot &:=& \neg\top\\
\end{array} \qquad
\begin{array}{lll}
\varphi_1 \lor \varphi_2 &:=& \neg(\neg\varphi_1 \land \neg\varphi_2) \\
\psi_1 \rightarrow \psi_2 &:=& \neg \psi_1 \lor \psi_2 \\
\end{array} \qquad
\begin{array}{lll}
\ltlFS \psi &:=& \ltlX\ltlF \psi \\
\ltlG \psi &:=& \neg\ltlF\neg \psi
\end{array}\]
We will also use $\ltlX^n\phi$ as shorthand for $\ltlX\cdots \ltlX\phi$ with $n$ occurrences of $\ltlX$, and similarly for 
$\ltlFS^n\phi$.

We will be interested in syntactic fragments of LTL.
For $O\subseteq \{\ltlU,\ltlF,\ltlFS,\ltlX,\land,\lor,\neg,\top,\bot\}$,
let $\LTL_O$ denote the fragment of LTL consisting of formulas built from
atomic propositions using only the connectives in $O$.
We write $\LTL_O[\ap]$ for the restriction of $\LTL_O$ to formulas over
the set of propositions $\ap$.

The semantics of LTL, over arbitrary ordinal words $w$, is given in Table~\ref{tab:ltl-semantics}.

\begin{table}[t]
\[
\begin{array}{@{}l@{\ }l@{\qquad}l@{\ }l@{}}
w \models \top & &

w\models \ltlX\varphi 
  & \text{iff  } |w|>1 \text{ and } w[1..] \models \varphi \\

w \models p 
  & \text{iff  } p \in w[0] &

w\models \ltlF\varphi 
  & \text{iff  } \exists \alpha<|w| \text{ s.t. } w[\alpha..] \models \varphi \\  

w\models \lnot \varphi 
  & \text{iff  } w\not\models \varphi &

w\models \varphi\, \ltlU\, \psi 
  & \text{iff  } \exists\gamma \text{ with } 0\le\gamma<|w| \text{ s.t. } w[\gamma..] \models \psi  \text{ and}\\  

w\models \varphi \land \psi 
  & \text{iff  } w\models \varphi 
    \text{ and } w\models \psi  & &
     \phantom{iff  } \forall \beta \text{ with } 0 \le \beta < \gamma,\ w[\beta..] \models \varphi \\
\end{array}
\]
\caption{Semantics of LTL over ordinal words}
\label{tab:ltl-semantics}
\end{table}

We will denote by $\sema{\phi}$ the set of all  ordinal words $w$ such that $w\models\phi$. We will also write $\semafin{\phi}$ for $\sema{\phi}\cap\finwords$, and, for all
ordinals $\alpha$, $\sema{\phi}^{\leq \alpha}$, 
$\sema{\phi}^{< \alpha}$ and $\sema{\phi}^{=\alpha}$ are defined similarly.
Note that, whether two LTL-formulas are equivalent depends on the class of words over which they are interpreted. This is illustrated by the following examples.

\begin{example}
\label{ex:XFvsFX}
 $\sema{\ltlF\ltlX p}^{\leq \omega}=\sema{\ltlX\ltlF p}^{\leq \omega}$
  but 
  $\sema{\ltlF\ltlX p}\neq \sema{\ltlX\ltlF p}$, as witnessed by the $\omega+\omega$-length word in Example~\ref{ex:ordinal-word}, which satisfies $\ltlX\ltlF p$ and not $\ltlF\ltlX p$.
\end{example}

\begin{example}
     Recall $\ltlG\phi$ is shorthand for 
     $\neg\ltlF\neg\phi$.
     $\semafin{\ltlG\ltlX p}=\semafin{\bot}$ but
     $\sema{\ltlG\ltlX p}^{=\omega}\neq\sema{\bot}^{=\omega}$.
\end{example}

On the other hand, for the fragment $\LTL_{\ltlFS,\land,\lor,\top,\bot}$, there is no difference between finite words and arbitrary ordinal words when it comes to equivalence:

\begin{restatable}{theorem}{thmFiniteInfiniteEquivalence}
\label{thm:finite-infinite-equivalence}
Let $\phi,\psi$ be
$\LTL_{\ltlFS,\land,\lor,\top,\bot}$-formulas. Then
$\sema{\phi}=\sema{\psi}$
iff
$\semafin{\phi}=\semafin{\psi}$.
The same holds for $\LTL_{\ltlX,\land,\neg}$.
\end{restatable}

For LTL over transfinite words, one can 
 restrict attention to finitely presentable words:

\begin{restatable}{theorem}{thmRegularExamples}
\label{thm:regular-examples}
Let $\phi,\psi$ be
LTL-formulas. If the set 
$\sema{\phi}\setminus\sema{\psi}$ is non-empty, 
it contains a finitely presentable word.
\end{restatable}

\section{Characterizing Formulas by Examples}

By a \emph{labeled example}
we will mean a pair $(w,lab)$
where $w$ is an ordinal word
and $lab\in\{+,-\}$. An
LTL-formula $\phi$ \emph{fits}
a labeled example $(w,lab)$ if
either $w\models\phi$ and $lab=+$, or $w\not\models\phi$
and $lab=-$. A set of labeled
examples $E$ \emph{uniquely characterizes} an LTL-formula $\phi$ relative to $\LTL_O[\ap]$ (for some set of operators $O$ and some set of atomic propositions $\ap$) if 
(i) $\phi$ fits every labeled example in $E$, and (ii) for every
$\LTL_O[\ap]$-formula $\psi$ that fits all labeled examples in $E$, we have that
$\sema{\psi}=\sema{\phi}$. Uniquely characterizing sets of examples are also known as \emph{teaching sets} in the literature. 
We say that
$\LTL_O[\ap]$ 
\emph{admits finite characterizations} if every 
$\phi\in\LTL_O[\ap]$ is uniquely
characterized, relative to 
$\LTL_O[\ap]$, by a finite set of labeled examples.

\begin{example}
Consider the formula $\phi=p\land \ltlX p$. This formula fits the
following  labeled examples: $(\{p\}\{p\},+)$,
 $(\emptyset\{p\},-)$, and 
 $(\{p\}\emptyset,-)$.
Moreover, $\phi$ is uniquely characterized,
relative to $\LTL_{\ltlX,\land}[\{p,q\}]$, by these labeled examples. On the other hand,
$\phi$ is \emph{not} uniquely characterized by these examples
relative to $\LTL_{\ltlX,\land,\lor}[\{p,q\}]$, since the formula
$\psi = \phi\lor q$ also fits and
$\sema{\phi}\neq\sema{\psi}$.
\end{example}

Observe that positive results are preserved under restricting the class of
formulas, whereas negative results are preserved under enlarging it. Thus,
every positive result for $\LTL_O[\ap]$ also holds for
$\LTL_{O'}[\ap']$ whenever $O'\subseteq O$ and $|\ap'|\le|\ap|$.
Conversely, every negative result for $\LTL_O[\ap]$ also holds for
$\LTL_{O'}[\ap']$ whenever $O\subseteq O'$ and $|\ap|\le|\ap'|$.

These notions can be further refined by restricting attention to words of certain lengths (e.g., finite words, or 
words of length at most $\omega$).  We say that
$\LTL_O[\ap]$ admits finite characterizations \emph{over $\words$} if for every 
$\phi\in\LTL_O[\ap]$ there is a
finite set of labeled examples $E$ consisting of  words in $\words$, such that  (i) $\phi$ fits every labeled example in $E$, and (ii) for every
$\LTL_O[\ap]$-formula $\psi$ that fits all labeled examples in $E$, we have that
$\sema{\psi}\cap\words=\sema{\phi}\cap\words$.

If an LTL-formula $\phi$ admits a uniquely characterizing set of  labeled examples $E$ with respect to a fragment, it means that \emph{every} consistent learner for the given fragment is guaranteed to output $\phi$ (or an equivalent formula) by including the examples in $E$ in the input. A weaker, learner-dependent notion of \emph{characteristic samples} will be discussed in Section~\ref{sec:characteristic-samples}.

\section{Results for LTL over finite words}
\label{sec:finite}

Fortin et al.~\cite{FortinEtAlKR2022UniqueTemporalInstanceQueries} previously studied the existence of finite unique characterizations for fragments of LTL over finite words. In particular,
they showed that $\LTL_{\ltlX,\ltlFS,\land,\top}[\ap]$ admits finite characterizations over $\finwords$. On the other hand, they showed that $\LTL_{\ltlF,\land}[\{p,q\}]$ and $\LTL_{\ltlU,\ltlX}[\{p\}]$ do not admit finite characterizations over $\finwords$. 
\footnote{A stronger positive result is stated in \cite{FortinEtAlKR2022UniqueTemporalInstanceQueries} pertaining to $\LTL_{\ltlFS,\ltlX,\land,\top,\bot}[\ap]$, but the proof is flawed and applies only to the fragment without $\bot$. Also, note that \cite{FortinEtAlKR2022UniqueTemporalInstanceQueries} uses a different variant of the Until operator. } 
They also consider subfragments consisting of \emph{path-shaped} formulas,
in which conjunctions of temporal statements (such as $\ltlF p\land \ltlF q$) are disallowed. For instance,
they showed that
even the path-shaped subfragment of  $\LTL_{\ltlF,\land}[\{p,q\}]$ does not admit finite characterizations.

\begin{example}[\cite{FortinEtAlKR2022UniqueTemporalInstanceQueries}]\label{ex:Fpandq}
    $\LTL_{\ltlF,\land}[\{p,q\}]$ does not admit finite characterizations over $\finwords$, and this holds already even just for path-shaped formulas. Specifically, the formula $\ltlF(p\land q)$ cannot be distinguished
    from all formulas of the form
    $\ltlF(p\land \ltlF(q\land \ltlF(p\land \ltlF(q\land \cdots))))$
    using only finitely many examples,
    where examples are finite words.
\end{example}

In~\cite{FortinEtAlKR2022UniqueTemporalInstanceQueries}, the authors 
subsequently perform an in-depth investigation into 
the existence of finite characterizations for 
various
path-shaped  subfragments of 
$\LTL_{\{\ltlF,\ltlFS,\ltlX,\land,\top,\bot\}}[\ap]$.
We will not consider such path-shaped fragments here, since they are quite restrictive, and we refer
the reader to~\cite{FortinEtAlKR2022UniqueTemporalInstanceQueries} for details.

The aforementioned results raise the question whether there may be other natural syntactic fragments of LTL that admit finite characterizations over $\finwords$. The following 
theorem provides an exhaustive answer to this question,
for fragments obtained by restricting the set of
allowed Boolean and temporal connectives. 
The proof is by an exhaustive case analysis.\looseness=-1

\begin{restatable}{theorem}{thmMainFinite}\label{thm:ltlfin}
Let $\ap$ be any finite set of atomic propositions with $|\ap|\geq 2$, and let $O\subseteq \allowbreak \{\ltlU,\ltlF,\ltlFS,\ltlX,\land,\lor,\neg,\top,\bot\}$ containing at least one temporal operator. Then  $\LTL_
        O[\ap]$ admits finite characterizations over $\finwords$ if and only if $O$ is a subset of one of the following sets:
            \begin{multicols}{3}
\begin{enumerate}
    \item $\{\ltlFS,\ltlX,\land,\top\}$ 
    \item $\{\ltlF,\lor,\top,\bot\}$  
    \item $\{\ltlF,\neg,\top,\bot\}$   
    \item $\{\ltlF,\ltlFS,\ltlX,\top\}$   
    \item $\{\ltlX,\neg\}$   
    \item $\{\ltlFS,\neg\}$   
\end{enumerate}
\end{multicols}

\end{restatable}

A few observations are worth making, based on the above theorem. First of all, any fragment of LTL
containing the $\ltlU$ operator fails to admit 
finite characterizations. Secondly, fragments of LTL that 
include all Boolean connectives (and at least one
temporal operator) do not admit finite characterizations. In fact, in order to admit
finite characterizations, a fragment of LTL must
omit conjunction or disjunction. However, omitting conjunction or disjunction is not enough: $\LTL_{\ltlF,\land}$ does not admit finite characterizations over $\finwords$ (as we already discussed) and the same holds for $\LTL_{\ltlFS,\lor}$ and 
 $\LTL_{\ltlX,\lor}$.

\begin{example}
    $\LTL_{\ltlX,\lor}[\{p\}]$ does not admit finite characterizations over $\finwords$. Indeed, the formula $p$ cannot be distinguished from all formulas of the form $p\lor\ltlX^n p$ ($n>0$) using only finitely many labeled examples, where the examples are finite words. A similar argument
    applies to $\LTL_{\ltlFS,\lor}[\{p\}]$.
\end{example}

As we will see in the next sections, 
omitting negation as a Boolean connective \emph{can} be enough
to obtain finite characterizability
once we go beyond finite words.

\section{Results for LTL over \texorpdfstring{$\omega$}{omega}-words}
\label{sec:omega}

We start by giving an example showing that 
infinite words can help when it comes to 
characterizing LTL-formulas. 

\begin{example}
We saw in \autoref{ex:Fpandq} that $\ltlF(p\land q)$ is
not finitely characterizable in
$\LTL_{\ltlF,\land}[\{p,q\}]$ over $\finwords$. This formula \emph{is}
finitely characterizable
in $\LTL_{\ltlF,\land,\lor,\top,\bot}[\{p,q\}]$ (over $\allwords$) by finitely presentable $\omega$-words, namely
 by the positive example $\emptyset~\{p,q\}~\emptyset^\omega$
 (or simply $\emptyset~\{p,q\}$)
and the negative example $(\{p\}~\{q\})^\omega$.
On the other hand, the formula
$\ltlF(p\land q\land \ltlF (r\land\ltlF(p\land q)))$ is
\emph{not} finitely characterizable by
$\omega$-words (cf.~the proof of Theorem~\ref{thm:omega-counterexample} below).
\end{example}

This may suggest that, by allowing $\omega$-words, we can circumvent the negative results from the previous section. This is partly true. For example, we have:

\begin{restatable}{theorem}{thmXpos}
Let $\ap$ be any finite set of atomic propositions.
$\LTL_{\ltlX,\land,\lor,\top,\bot}[\ap]$ admits finite characterizations over $\words^{\leq\omega}$ and over $\words^{=\omega}$. In fact, every
$\LTL_{\ltlX,\land,\lor,\top,\bot}[\ap]$-formula is uniquely characterized with respect to 
$\LTL_{\ltlX,\land,\lor,\top,\bot}[\ap]$ (over $\allwords$)
by a finite set of finitely presentable labeled examples of ordinal length at most $\omega$.
\end{restatable}

The proof uses a monotonicity argument to show that, with finitely many labeled examples, we can force any fitting formula to depend only on a finite fixed-length prefix.

Unfortunately, other fragments of LTL still do not admit finite characterizations, even when $\omega$-words are allowed. In particular, 

\begin{restatable}{theorem}{thmomegacounterexample}
\label{thm:omega-counterexample}
$\LTL_{\ltlF,\land}[\{p,q,r\}]$ does not admit finite characterizations over $\words^{\leq \omega}$ or $\words^{=\omega}$. \end{restatable}

Since $\ltlF\phi$ is definable
as $\phi\lor\ltlFS\phi$, the same holds for $\LTL_{\ltlFS,\land,\lor}[\{p,q,r\}]$.
We will refrain from  providing a complete classification in the style of Theorem~\ref{thm:ltlfin}. Instead, we will next
move on to the setting where examples can be transfinite words --- a setting where, as we will see, more positive results hold.

\section{Results for LTL over transfinite words}
\label{sec:transfinite}

Our main result, in this section, is:

\begin{restatable}{theorem}{thmMain}
\label{thm:main}
Let $\ap$ be any finite set of atomic propositions.
Then $\LTL_{\ltlFS,\land,\lor,\top,\bot}[\ap]$ admits finite characterizations (over $\allwords$).
In fact, every formula in this fragment is uniquely characterized by a finite
set of finitely presentable labeled examples of  length less than $\omega^2$.
\end{restatable}

Note that, since $\ltlF\phi$ is definable
as $\phi\lor\ltlFS\phi$, the result extends
immediately to $\LTL_{\ltlF,\ltlFS,\land,\lor,\top,\bot}$.

\begin{example}\label{ex:transfinite-characterization}
Consider again the formula $\ltlF(p\land q\land \ltlF (r\land \ltlF(p\land q)))$. We saw earlier that this formula cannot be uniquely characterized
 w.r.t. $\LTL_{\ltlF,\land}[\{p,q,r\}]$ over $\words^{\leq \omega}$ by finitely many labeled
 examples (of length at most $\omega$). 
 The above theorem tells us that the same formula \emph{can} be 
 uniquely characterized using transfinite examples. Indeed, it is uniquely characterized by
$S = \bigl(E^+ \times \{+\},\; E^- \times \{-\}\bigr)$,
where

\smallskip 
\hspace{-6mm}
    \begin{tabular}{@{}l@{}l@{}}
        $\begin{aligned}
        E^+ = \{&
        \{p,q\}\{r\}\{p,q\},\;
        \{p,q\}\{p,q,r\},\;
        \{p,q,r\},\\
        &\emptyset\{p,q\}\{r\}\{p,q\},\;
        \emptyset\{p,q\}\{p,q,r\},\;
        \emptyset\{p,q,r\}
        \}, \\ \phantom{foo}
        \end{aligned}$
        &
        $\begin{aligned}
        E^- = \{&
        \{p,q\}^{\omega}\cdot(\{p,r\}\{q,r\})^{\omega},\\
        &(\{p,r\}\{q,r\})^{\omega}\cdot\{p,q\}^{\omega}\cdot(\{p,r\}\{q,r\})^\omega,\\
        &(\{q,r\}\{p,r\})^{\omega}\cdot\{p,q\}^{\omega}\cdot(\{p,r\}\{q,r\})^\omega
        \}.
        \end{aligned}$    
    \end{tabular}
\end{example}

The examples above illustrate why words of length below $\omega^2$ suffice.
Intuitively, each level of nesting of the $\ltlF$ operator requires at most one
additional $\omega$-segment in the negative examples. Since every formula has finite
nesting depth, finitely many $\omega$-segments always suffice.

Theorem~\ref{thm:main} is optimal, as finite characterizations are lost when negation is added to the fragment.

\begin{restatable}{theorem}{thmXFFSneg}
$\LTL_{\ltlX,\wedge,\neg}[\{p\}]$,
$\LTL_{\ltlF,\wedge,\neg}[\{p\}]$ and
$\LTL_{\ltlFS,\wedge,\neg}[\{p\}]$ 
do not admit finite characteri\-zations over $\allwords$.
\end{restatable}

\medskip
\noindent
Moreover, extending the fragment $\LTL_{\ltlF,\ltlFS,\wedge,\vee, \top, \bot}$ with the $\ltlX$ operator also destroys finite characterizability.

\begin{restatable}{theorem}{thmFXandOr}
\label{thm:FX-and-or-no-finite-char}
$\LTL_{\ltlF,\ltlX,\wedge,\vee}[\ap]$
does not admit finite characterizations over $\allwords$.
\end{restatable}

In the remainder of this section, we 
describe the proof of Theorem~\ref{thm:main}.
It relies on viewing formulas through
their minimal witnesses and how these witnesses propagate across words.
To formalize this idea, we introduce a notion of \emph{structural embedding}
between words, captured by an order-preserving, label-monotone \emph{homomorphism}.

\begin{definition}[Homomorphism]
A \emph{homomorphism} from an $\alpha$-word $w$ to a $\beta$-word $w'$ is a function $h : \alpha \to \beta$ such that:
\begin{enumerate}
    \item \label{cond1} $h(0) = 0$;
    \item \label{cond2} $\forall i < \alpha$, $w[i] \subseteq w'[h(i)]$; and
    \item \label{cond3} $\forall i < j < \alpha$, $h(i) < h(j)$;
\end{enumerate}
\end{definition}
When such a homomorphism exists, we write $w \leq_{\mathrm{hom}} w'$. If condition~(\ref{cond1})  is dropped, we say that $h$ is a \emph{weak homomorphism} and
write $w \leq_{\mathrm{hom}^+} w'$.
For a set of words $E$, we denote by 
$E^\uparrow$ the \emph{upward closure} of $E$, i.e.,
$\set{\, w \mid \exists\, e \in E : e \leq_{\hom} w}$.

A central property of the fragment is that satisfaction is preserved
under such embeddings. Intuitively, enlarging labels and stretching
positions cannot destroy the truth of formulas built from
$\ltlF$, $\ltlFS$, and Boolean connectives.

\begin{restatable}[Monotonicity w.r.t.\ $\leq_{\hom}$]{lemma}{monotone}\label{lem:monotonicity}
For every $\varphi \in \LTL_{\ltlFS,\land,\lor,\top,\bot}[\ap]$
and ordinal words $w,w'$,
if $w \models \varphi$ and $w \leq_{\hom} w'$, then $w' \models \varphi$.
\end{restatable}

Monotonicity suggests that formulas can be characterized by their
minimal satisfying words. We now show that each formula admits a finite
set of such minimal witnesses, from which all satisfying words can be
obtained via homomorphic extension.

\begin{restatable}[Canonical set]{lemma}{canonicalset}\label{lem:canonical-set}
For every
$\varphi \in \LTL_{\ltlFS,\land,\lor,\top,\bot}[\ap]$,
there exists a finite set \(E^+_\varphi \subseteq (2^{\ap})^{<\omega}\) such that
$\sema{\varphi}
=
\{w \;| \; \exists e\in E^+_\varphi, e \le_{hom} w\} $.
We call \(E^+_\varphi\) the \emph{canonical set} of \(\varphi\).
\end{restatable}

The proof of the above lemma is based on an inductive 
construction. The interesting case is for conjunctions, where $E^+_{\phi\land\psi}$
must consist of homomorphically minimal examples satisfying both $\phi$ and $\psi$. Such a set can be constructed by taking, for each pair of words $w_1\in E^+_\phi$ and $w_2\in E^+_\psi$, all possible interleaving of $w_1$ and $w_2$. 

While the canonical set captures all positive behavior of a formula,
it does not by itself yield a characterization: one must also ensure
that words not satisfying the formula are excluded. To this end, we
construct a finite family of maximal negative examples that block all
possible embeddings of positive witnesses.

\begin{restatable}[Existence of maximal negative examples]{lemma}{maxnegexamples}\label{lem:max-negative-examples}
For every
$\varphi \in \LTL_{\ltlFS,\land, \lor,\top, \bot}[\ap]$,
there exists a \emph{finite} set \(E^-_\varphi \subseteq (2^{\ap})^{<\omega^2}\) such that
\[
\forall w \in (2^{\ap})^{\leq \omega}\;
\bigl( w \not\models \varphi \Rightarrow \exists f \in E^-_\varphi : w \leq_{\hom} f \bigr),
\qquad\text{and}\qquad
\forall f \in E^-_\varphi,\; f \not\models \varphi.
\]
We call \(E^-_\varphi\) a set of \emph{maximal negative examples} for \(\varphi\).
\end{restatable}

Lemma~\ref{lem:max-negative-examples} is proved by dualizing the canonical set $E^+_\phi$.
More precisely, we build a finite set of words that systematically block all possible
homomorphic embeddings of examples in $E^+_\phi$.
A key point is that blocking an embedding requires accounting for all
possible positions at which the letters of a positive witness may be
matched. For instance, preventing an embedding of
$e = \{p\}\{q\}\{p\}$ requires considering all cases in which $\{q\}$
is matched after an arbitrary finite prefix following $\{p\}$, and then
ensuring that $\{p\}$ does not appear thereafter. Since the length of the
prefix is unbounded, the construction must range over all finite lengths,
and then enforce an additional infinite constraint. This can be done by a procedure that yields words of
length up to $\omega^2$. \looseness=-1

Combining the canonical positive examples with the maximal negative
examples yields a finite labeled sample that uniquely characterizes the formula, thereby proving
Theorem~\ref{thm:main}.

By a further dualization argument we also obtain an analogous result for 
$\LTL_{\ltlGS,\land,\lor,\top,\bot}$:

\begin{restatable}{corollary}{corDual}
    Let $\ap$ be any finite set of proposition letters. $\LTL_{\ltlGS,\land,\lor,\top,\bot}[\ap]$ admits finite characterizations.
\end{restatable}

\begin{remark}\label{rem:bound}
We can bound the size of the finite characterizations produced by our constructions.
Let \(m := |\varphi|\) denote the size of the syntax tree of \(\varphi\), and let
\(h := |\ap|\).
The construction of Lemmas~\ref{lem:canonical-set} and~\ref{lem:max-negative-examples}
yields a finite characterization
$S = (E^+_\varphi \times \set{+},\, E^-_\varphi \times \set{-})$
such that
$|E^+_\varphi| \le 2^{O(m)}$ and $|E^-_\varphi| \le 2^{O(m\,2^m)}$.
Furthermore, each word is given by a word expression of size at most $O(mh2^{m+h})$. 
Thus, the overall size of the characterization is doubly exponential in $|\varphi|$.
Proof details appear in the appendix.
These bounds should be compared with the lower-bound results of
Fortin et al.~\cite{FortinEtAlKR2022UniqueTemporalInstanceQueries}, who showed that
\(\LTL_{\land,\ltlFS}\) does not admit polynomial characterizations over finite words.
\end{remark}

\section{Back to the Finite: Schematic Examples}\label{sec:schematic-examples}

We turn to define \emph{schematic examples}, a concise way to represent (possibly) infinitely many finite examples.
Consider a set of atomic propositions $\ap$. 
We use $\mathbb{B}(\ap)$ to denote the Boolean expressions over $\ap$. Note that a Boolean expression $\texttt{b}\in\mathbb{B}(\ap)$ defines a set of letters over $\Sigma=2^{\ap}$, the alphabet on which models of LTL are defined.

A \emph{schematic example} is specified by a union-free regular expression $r$ over the alphabet $\mathbb{B}(\ap)$ of star-height at most 1 (that is, there are no nested occurrences of the Kleene star).%
We say that 
$\phi$ \emph{fits} $(r,+)$ 
if $L(r)\subseteq\sema{\phi}$,
and that $\phi$ fits $(r,-)$
if $L(r)\cap\sema{\phi}=\emptyset$.
A schematic example can be viewed as an underspecified representation of a trace. This underspecification arises at two levels. First, each letter is a Boolean formula over the atomic propositions and therefore represents a set of possible valuations rather than a single, fully specified state. Second, Kleene star leaves the number of repetitions of a finite pattern unspecified, thereby abstracting from the exact duration of a phase or the number of times it occurs. We exclude union and restrict the star height to one, ensuring that every repeated subexpression is itself a fixed finite pattern containing no further unbounded repetition. Consequently, a schematic example can be read as a linear succession of finitely described phases, some of variable duration.

\begin{example}
    Fix $\ap=\{p,q\}$. 
    The LTL-formula $\ltlF (p\land q)$ is uniquely characterized w.r.t. $\LTL[\{p,q\}]$ over $\finwords$ by the set of labeled schematic examples 
    \[E=\{(\true^* \cdot (p \wedge q) \cdot\true^*),+),\ \ ((\neg p \vee \neg q)^*,-)\} \qquad\qquad \text{  where  } 
    \true = p \vee \neg p. 
    \]
\end{example}

Building on the results about transfinite examples in the previous section, we can  show:

\begin{restatable}{theorem}{thmFinitecharacterizationSimpleSchematic}\label{thm:finite-char-simple-schematic}
$\LTL_{\ltlFS,\land,\lor,\top,\bot}[\ap]$ admits finite characterizations over $\finwords$ by schematic examples.
\end{restatable}

\begin{proof}[Proof sketch]
The construction proceeds by reducing finite characterizations consisting of regular (i.e.,~finitely presented) transfinite words
to characterizations consisting of schematic examples.

We already showed that every formula \(\varphi\) in the fragment admits a finite
characterization by a sample \(S\) consisting of positively labeled finite
words and negatively labeled regular transfinite words.
We transform this sample into a  schematic sample
\(S' := Schematic(S)\) by replacing each word \(w\) with its schematic
representation \(r(w)\).

The key observation is that, for this fragment, satisfaction is preserved
under replacing \(\omega\)-iterations by finite repetitions.
Thus each schematic expression \(r(w)\) faithfully captures the behavior
of the original word \(w\) with respect to \(\varphi\).
Consequently, \(S'\) fits \(\varphi\), and moreover, any formula that fits
\(S'\) must also fit \(S\), and hence is equivalent to \(\varphi\).

Finally, by construction, every schematic expression in \(S'\) is obtained
by replacing \(\omega\)-powers with Kleene-star, and therefore has
star-height \(1\).
\end{proof}

\begin{example}[Schematic examples]\label{ex:schematic-examples}
Recall the formula
$\varphi = \ltlF\bigl(p\land q \land \ltlF(r\land \ltlF(p\land q))\bigr)$
from Example~\ref{ex:transfinite-characterization}. We now give the corresponding schematic examples.  %

The positive examples can be grouped into two schematic examples:
\[
R^+ = \{ \true^* \cdot (p {\land} q {\land} r),\quad
          \true^* \cdot (p {\land} q) \cdot r \cdot (p {\land} q) \}.
\]
The negative examples can be grouped into the single schematic example
\[
R^-=\{ (\neg(p{\land} q))^*\cdot (\neg r)^* \cdot (\neg (p{\land} q))^* \}.\]
\end{example}

Using $\mathbb{B}(\ap)$ as the basis of regular expressions is common in practice; in particular, such expressions are part of the IEEE standards for temporal logics PSL~\cite{EisnerF06,EisnerF18} and SVA~\cite{CernyDHK2014}. Furthermore, using the fusion operator of PSL, that concatenates two regular expressions with one letter overlap, the two positive schematic examples above can be merged into the single expression
$\true^*\circ (p{\land} q)\circ\true^*\circ(r)\circ\true^*\circ(p{\land} q)$.

It is well known that, for every LTL formula $\phi$, $\semafin{\phi}$ and $\semafin{\neg\phi}$  are regular languages.
A classic result in~\cite{Nagy2004NormalFormRegex} states that every regular language is a union of finitely many languages definable by union-free regular expressions. The union-free regular expressions in question may, however, have unbounded star
height. Hence, this observation alone does not yield finite unique characterizations consisting of schematic examples. 

The results in this section are restricted to LTL interpreted over finite words. For instance, $\ltlF p$ and $\ltlF (p\land\neg \ltlFS p)$ are non-equivalent in general, but cannot be distinguished by finite words, %
and therefore also not by schematic examples. 
In this sense, our results that use transfinite examples still offer benefits
beyond those of schematic examples.
Even 
over $\finwords$, we believe that the full language of LTL does not admit finite characterizations  by  schematic examples, but this remains to be proved. %

\section{Another notion of ``good'' examples:
characteristic samples}
\label{sec:characteristic-samples}

The main focus of this paper is on the existence of \emph{unique characterizations} for fragments of LTL. 
In this section, we briefly compare this notion with another notion of a ``good'' set of examples, namely \emph{characteristic samples}. Let 
a \emph{concept class} be a triple $\class{C}=(C,\mathcal{X},\sema{\cdot})$, where $C$ is a set (the \emph{concepts}), $\mathcal{X}$ is a set (the \emph{examples}) and where, for each $c\in C$, $\sema{c}\subseteq\mathcal{X}$.

\newcommand{\wpfin}{\wp_{\text{finite}}}

\begin{definition}[Learner and teacher]\quad
Fix a concept class $\class{C}=(C,\mathcal{X},\sema{\cdot})$.
\begin{itemize}
    \item 
A \emph{learner} for $\class{C}$ is a function $\alg{L}:\wpfin(\mathcal{X}\times\{0,1\})\to\Psi$ with the property that $\alg{L}(S)$ fits $S$, assuming $S$ fits at least one element of $\class{C}$.\footnote{If $S$ fits no element of $\class{C}$ then there is no requirement on $\alg{L}(S)$.}
\item
A \emph{teacher} for $\class{C}$ is a function $\alg{T}:\Psi\to \wpfin(\mathcal{X}\times\{0,1\})$ such that $\alg{T}(\psi)$ fits $\psi$.
\end{itemize}
Note that learners and teachers need not be computable.
\end{definition}

\begin{definition}[Characteristic Samples]\quad
\begin{itemize}
    \item 

    A sample $S$ is a \emph{characteristic sample} for $\psi$ and a learner $\alg{L}$
if $S$ fits $\psi$ and for every
sample $S'\supseteq S$ that fits $\psi$ we have
$\sema{\alg{L}(S')} = \sema{\psi}$.
The intuition is that additional information consistent with $\psi$ beyond $S$ will not cause the learner to change its mind about the correct concept.

\item
A concept class $\class{C}$ \emph{has characteristic samples for a learner} $\alg{L}$ if there exists a teacher $\alg{T}$ such that for every $\psi$ in $\class{C}$, $\alg{T}(\psi)$ is a characteristic sample for $\psi$ and $\alg{L}$.
A concept class $\class{C}$ \emph{has characteristic samples} if it has characteristic samples for some learner.
\end{itemize}
\end{definition}

The following is a known property of characteristic samples.

\begin{lemma}
    Let $\class{C}=(\Psi,\mathcal{X},\sema{\cdot})$ be a concept class. Assume $\class{C}$ has characteristic samples with respect to learner $\alg{L}$ and teacher $\alg{T}$. Let $\psi,\psi'\in\Psi$.
    If $\sema{\psi}\neq\sema{\psi'}$ then there exists $(x,b)\in\alg{T}(\psi)\cup \alg{T}(\psi')$ such that $x\in\sema{\psi}\oplus\sema{\psi'}$ where $\oplus$ denotes the symmetric difference.
\end{lemma}

\begin{proof}
Assume this is not the case. Let $S$ be a sample subsuming $\alg{T}(\psi)\cup \alg{T}(\psi')$. Let $\psi''=\alg{L}(S)$.
Since $S$ fits $\psi$ and subsumes $\alg{T}(\psi)$, by the definition of characteristic samples, $\psi''$ should satisfy $\sema{\psi''}=\sema{\psi}$. The same reasoning entails that $\sema{\psi''}=\sema{\psi'}$, contradicting $\sema{\psi}\neq\sema{\psi'}$.
\end{proof}

The notion of unique characterization is stronger than that of characteristic samples, as formally stated in the following theorem.

\begin{restatable}{theorem}{thmUniqueImpliesCharacteristic}
    If a concept class $\class{C}$ admits a unique characterization, then it has characteris\-tic samples.
\end{restatable}

It follows that all fragments for which we have seen positive results regarding unique characterizations also have characteristic samples. The converse is not true, since as the following theorem shows the concept class consisting of all LTL formulas has characteristic samples, while we have provided several negative results regarding unique characterization for fragments of LTL.

\begin{restatable}{theorem}{thmLTLcharsamples}
\label{thm:LTL-char-samples}
 Let $\ap$ be any finite set of atomic propositions. 
  $\LTL[\ap]$ admits characteristic samples over $\finwords$. 
\end{restatable}

The theorem can be proved using Gold's enumeration technique~\cite{Gold67}.

This stands in sharp contrast with the results in Section~\ref{sec:finite}, which showed that finite unique characterizations over $\finwords$ rarely exist.

In a sense, as evident from the respective definitions and proofs, unique characterization can be viewed as a strengthening of the notion of characteristic samples, requiring the property to hold with respect to any learner rather than a particular one.

Another aspect in which the two notions can be compared is the size of the example set. We measure the size of a set as the sum of the lengths of the examples it contains. Gold's method, used in the proof of \autoref{thm:LTL-char-samples}, yields characteristic samples of exponential cardinality (that is, the number of examples is exponential, regardless of their size). The result of Barzdin and Freivalds~\cite{BarzdinF72} shows that characteristic samples of polynomial cardinality are possible. This result, however, also does not take into account the size of the examples.

Regarding the size of the examples themselves, we can obtain a doubly exponential upper bound by translating the formula into an automaton and taking a characteristic sample for it. An LTL formula of size $n$ can be translated into a B\"uchi automaton $B_\psi$ of size $2^{O(n)}$. The language $L$ of a B\"uchi automaton can be represented as the following language of finite words
$(L)_{\$}=\{ u \$ v \mid u(v)^\omega \in L\}$
since  two $\omega$-regular languages are equivalent iff they agree on the set of ultimately periodic. Calbrix et al.~\cite{CalbrixNP93} showed that $(L)_{\$}$ is a regular language of finite words.  Kuperberg et al.~\cite{KuperbergPP19} showed that a non-deterministic B\"uchi automaton with $n$ states can be converted to a DFA with at most $2^n+3^{n^2}$. Thus, given an LTL formula $\psi$ of size $n$, we can obtain a DFA for $(L(B_\psi))_{\$}$ whose size is doubly exponential in $n$. Since a characteristic sample for a DFA has size polynomial in the number of its states, the overall doubly exponential upper bound follows.
We complement this by a matching lower bound.

\begin{restatable}{proposition}{propDFAlowerbound}
    There exists a family of languages $\{L_n\}_{n\in\mathbb{N}}$ that can be represented by LTL formulas of size linear in $n$ whereas the DFA for $(L_n)_\$$  %
    requires at least $2^{2^n}$ states.
\end{restatable}

This opens the question whether working with some other representation for $\omega$-regular languages may yield a smaller size automaton. We note that the answer is negative for any deterministic $\omega$-automaton (resp. family of DFAs used to represent $\omega$-regular languages~\cite{MalerS93,AngluinF16,FismanGZ24}), since the same proof as above goes through (resp. for the leading automaton).

Thus, the doubly-exponential bound matches the complexity of obtaining a characteristic sample for LTL via a corresponding deterministic automaton. It remains possible, however, that more efficient constructions of characteristic samples for LTL formulas exist that do not proceed through such automata representations.

\section{Conclusion}

We studied when LTL-formulas can be uniquely characterized by finite sets of labeled examples. We gave a complete classification, over finite words, of the operator fragments that admit such finite unique characterizations, showing that this property is quite limited in this setting. We then showed that moving to richer example domains changes the picture substantially: while $\omega$-words already help for some fragments, transfinite words enable finite unique characterizations for large monotone fragments of LTL. Finally, returning to the finite setting, we introduced schematic examples as compact representations of families of finite words, enabling unique characterization results for large monotone fragments of LTL. 

Together, these results provide a systematic account of how the descriptive power of examples for LTL-formulas depends on the type of examples that is allowed. They also suggest that schematic examples are a promising middle ground between concrete traces and more abstract symbolic descriptions: they remain close to the kinds of finite behaviors practitioners use, while being quite expressive in terms of describing the behavior of an LTL-formula. \looseness=-1

Several directions remain open. On the technical side, it would be natural to extend the present analysis to larger fragments involving both $\ltlF$ and $\ltlG$. Furthermore, restricting the set of allowed operators is a natural first step, but one could also systematically study other kinds of syntactic fragments (such as obtained by forbidding specific operator nesting patterns, as is done in some of the results in~\cite{FortinEtAlKR2022UniqueTemporalInstanceQueries}). 
On the conceptual side, it would be interesting to revisit related notions such as characteristic samples, active learning, as well as interfaces for debugging or explanation, through the lens of richer example formats, and to better understand when compact, human-interpretable examples can serve as faithful descriptions of temporal specifications.
\bibliographystyle{plainurl}
\bibliography{bib.bib}

@article{CateFOS26arxiv,
  author       = {Balder ten Cate and
                  Dana Fisman and
                  Roi Ohayon and
                  Patrik Sestic},
  title        = {Characterizing {LTL} Formulas by Examples},
  journal      = {CoRR},
  volume       = {abs/2604.22097},
  year         = {2026},
  url          = {https://doi.org/10.48550/arXiv.2604.22097},
  doi          = {10.48550/ARXIV.2604.22097},
  eprinttype   = {arXiv},
  eprint       = {2604.22097},
  bibsource    = {dblp computer science bibliography, https://dblp.org}
}

@article{BalderCQ,
author = {ten Cate, Balder and Dalmau, V\'{i}ctor},
title = {Conjunctive Queries: Unique Characterizations and Exact Learnability},
year = {2022},
issue_date = {December 2022},
publisher = {Association for Computing Machinery},
address = {New York, NY, USA},
volume = {47},
number = {4},
issn = {0362-5915},
url = {https://doi.org/10.1145/3559756},
doi = {10.1145/3559756},
journal = {ACM Trans. Database Syst.},
month = {nov},
articleno = {14},
numpages = {41}
}

@inproceedings{DBLP:conf/ijcai/FunkJL22,
  author       = {Maurice Funk and
                  Jean Christoph Jung and
                  Carsten Lutz},
  title        = {Frontiers and Exact Learning of {ELI} Queries under {DL}-{L}ite Ontologies},
  booktitle    = {{IJCAI}},
  pages        = {2627--2633},
  publisher    = {ijcai.org},
  year         = {2022}
}

@inproceedings{Nagy2004NormalFormRegex,
  author    = {Benedek Nagy},
  title     = {A Normal Form for Regular Expressions},
  booktitle = {Supplemental Papers for DLT'04},
  editor    = {Calude, C. S. and Calude, E. and Dinnen, M. J.},
  publisher = {CDMTCS},
  address   = {Auckland, New Zealand},
  year      = {2004}
}

@inproceedings{CalbrixNP93,
  author       = {Hugues Calbrix and
                  Maurice Nivat and
                  Andreas Podelski},
  editor       = {Stephen D. Brookes and
                  Michael G. Main and
                  Austin Melton and
                  Michael W. Mislove and
                  David A. Schmidt},
  title        = {Ultimately Periodic Words of Rational \emph{w}-Languages},
  booktitle    = {Mathematical Foundations of Programming Semantics, 9th International
                  Conference, New Orleans, LA, USA, April 7-10, 1993, Proceedings},
  series       = {Lecture Notes in Computer Science},
  volume       = {802},
  pages        = {554--566},
  publisher    = {Springer},
  year         = {1993},
  url          = {https://doi.org/10.1007/3-540-58027-1\_27},
  doi          = {10.1007/3-540-58027-1\_27},
  bibsource    = {dblp computer science bibliography, https://dblp.org}
}

@inproceedings{ChandraS76,
  author       = {Ashok K. Chandra and
                  Larry J. Stockmeyer},
  title        = {Alternation},
  booktitle    = {17th Annual Symposium on Foundations of Computer Science, Houston,
                  Texas, USA, 25-27 October 1976},
  pages        = {98--108},
  publisher    = {{IEEE} Computer Society},
  year         = {1976},
  url          = {https://doi.org/10.1109/SFCS.1976.4},
  doi          = {10.1109/SFCS.1976.4},
  bibsource    = {dblp computer science bibliography, https://dblp.org}
}

@inproceedings{MalerS93,
  author       = {Oded Maler and
                  Ludwig Staiger},
  editor       = {Patrice Enjalbert and
                  Alain Finkel and
                  Klaus W. Wagner},
  title        = {On Syntactic Congruences for Omega-Languages},
  booktitle    = {{STACS} 93, 10th Annual Symposium on Theoretical Aspects of Computer
                  Science, W{\"{u}}rzburg, Germany, February 25-27, 1993, Proceedings},
  series       = {Lecture Notes in Computer Science},
  volume       = {665},
  pages        = {586--594},
  publisher    = {Springer},
  year         = {1993},
  url          = {https://doi.org/10.1007/3-540-56503-5\_58},
  doi          = {10.1007/3-540-56503-5\_58},
  bibsource    = {dblp computer science bibliography, https://dblp.org}
}

@article{AngluinF16,
  author       = {Dana Angluin and
                  Dana Fisman},
  title        = {Learning regular omega languages},
  journal      = {Theor. Comput. Sci.},
  volume       = {650},
  pages        = {57--72},
  year         = {2016},
  url          = {https://doi.org/10.1016/j.tcs.2016.07.031},
  doi          = {10.1016/J.TCS.2016.07.031},
  bibsource    = {dblp computer science bibliography, https://dblp.org}
}

@inproceedings{FismanGZ24,
  author       = {Dana Fisman and
                  Emmanuel Goldberg and
                  Oded Zimerman},
  editor       = {Rastislav Kr{\'{a}}lovic and
                  Anton{\'{\i}}n Kucera},
  title        = {A Robust Measure on FDFAs Following Duo-Normalized Acceptance},
  booktitle    = {49th International Symposium on Mathematical Foundations of Computer
                  Science, {MFCS} 2024, Bratislava, Slovakia, August 26-30, 2024},
  series       = {LIPIcs},
  volume       = {306},
  pages        = {53:1--53:17},
  publisher    = {Schloss Dagstuhl - Leibniz-Zentrum f{\"{u}}r Informatik},
  year         = {2024},
  url          = {https://doi.org/10.4230/LIPIcs.MFCS.2024.53},
  doi          = {10.4230/LIPICS.MFCS.2024.53},
  bibsource    = {dblp computer science bibliography, https://dblp.org}
}

@article{Gold67,
	author    = {E.M. Gold},
	title     = {Language Identification in the Limit},
	journal   = {Information and Control},
	volume    = {10},
	number    = {5},
	pages     = {447--474},
	year      = {1967},
}

@article{BarzdinF72,
author={Barzdin, J. M. and Freivalds, R. V.},
title={On the prediction of total recursive functions},
journal={Doklady Akademii Nauk SSSR},
volume={206},
number={3},
year={1972},
pages={521--524},
note={(In Russian)}
}

@inproceedings{KuperbergPP19,
  author       = {Denis Kuperberg and
                  Laureline Pinault and
                  Damien Pous},
  editor       = {Piotrek Hofman and
                  Michal Skrzypczak},
  title        = {Coinductive Algorithms for {B}{\"{u}}chi Automata},
  booktitle    = {Developments in Language Theory - 23rd International Conference, {DLT}
                  2019, Warsaw, Poland, August 5-9, 2019, Proceedings},
  series       = {Lecture Notes in Computer Science},
  volume       = {11647},
  pages        = {206--220},
  publisher    = {Springer},
  year         = {2019},
  url          = {https://doi.org/10.1007/978-3-030-24886-4\_15},
  doi          = {10.1007/978-3-030-24886-4\_15},
  bibsource    = {dblp computer science bibliography, https://dblp.org}
}

@inproceedings{Vardi95,
  author       = {Moshe Y. Vardi},
  editor       = {Faron Moller and
                  Graham M. Birtwistle},
  title        = {An Automata-Theoretic Approach to Linear Temporal Logic},
  booktitle    = {Logics for Concurrency - Structure versus Automata (8th Banff Higher
                  Order Workshop, Banff, Canada, August 27 - September 3, 1995, Proceedings)},
  series       = {Lecture Notes in Computer Science},
  pages        = {238--266},
  publisher    = {Springer},
  year         = {1995},
  url          = {https://doi.org/10.1007/3-540-60915-6\_6},
  doi          = {10.1007/3-540-60915-6\_6},
  bibsource    = {dblp computer science bibliography, https://dblp.org}
}

@book{EisnerF06,
  author       = {Cindy Eisner and
                  Dana Fisman},
  title        = {A Practical Introduction to {PSL}},
  series       = {Series on Integrated Circuits and Systems},
  publisher    = {Springer},
  year         = {2006},
  url          = {https://doi.org/10.1007/978-0-387-36123-9},
  doi          = {10.1007/978-0-387-36123-9},
  isbn         = {978-0-387-35313-5},
  bibsource    = {dblp computer science bibliography, https://dblp.org}
}

@incollection{EisnerF18,
  author       = {Cindy Eisner and
                  Dana Fisman},
  editor       = {Edmund M. Clarke and
                  Thomas A. Henzinger and
                  Helmut Veith and
                  Roderick Bloem},
  title        = {Functional Specification of Hardware via Temporal Logic},
  booktitle    = {Handbook of Model Checking},
  pages        = {795--829},
  publisher    = {Springer},
  year         = {2018},
  url          = {https://doi.org/10.1007/978-3-319-10575-8\_24},
  doi          = {10.1007/978-3-319-10575-8\_24},
  bibsource    = {dblp computer science bibliography, https://dblp.org}
}

@book{CernyDHK2014,
  title     = {SVA: The Power of Assertions in SystemVerilog},
  author    = {Eduard Cerny and Surrendra Dudani and John Havlicek and Dmitry Korchemny},
  publisher = {Springer},
  year      = {2014},
  isbn      = {978-3-319-07138-1},
  date      = {2014-08-24},
  pages     = {590}
}

@article{DemriRabinovich2010,
    title      = {The complexity of linear-time temporal logic over the class of ordinals},
    author     = {Stephane Demri and Alexander Rabinovich},
    url        = {https://lmcs.episciences.org/1230},
    doi        = {10.2168/LMCS-6(4:9)2010},
    journal    = {Logical Methods in Computer Science},
    issn       = {1860-5974},
    volume     = {Volume 6, Issue 4},
    eid        = 9,
    year       = {2010},
    month      = {Dec},
}

@article{DemriNowak2006,
author = {Demri, St\'{e}phane and Nowak, David},
title = {Reasoning about Transfinite Sequences},
journal = {International Journal of Foundations of Computer Science},
volume = {18},
number = {01},
pages = {87-112},
year = {2007},
doi = {10.1142/S0129054107004589},
URL = { 
        https://doi.org/10.1142/S0129054107004589
},
eprint = { 
        https://doi.org/10.1142/S0129054107004589
}
}

@book{BuchiSiefkes1973,
  author    = {Julius Richard B{\"u}chi and Dirk Siefkes},
  title     = {The Monadic Second-Order Theory of All Countable Ordinals},
  series    = {Lecture Notes in Mathematics},
  volume    = {328},
  publisher = {Springer},
  year      = {1973}
}

@inproceedings{XuChenTopcu20,
  author       = {Zhe Xu and
                  Yuxin Chen and
                  Ufuk Topcu},
  title        = {Adaptive Teaching of Temporal Logic Formulas to Preference-based Learners},
  booktitle    = {Thirty-Fifth {AAAI} Conference on Artificial Intelligence, {AAAI}
                  2021, Thirty-Third Conference on Innovative Applications of Artificial
                  Intelligence, {IAAI} 2021, The Eleventh Symposium on Educational Advances
                  in Artificial Intelligence, {EAAI} 2021, Virtual Event, February 2-9,
                  2021},
  pages        = {5061--5068},
  publisher    = {{AAAI} Press},
  year         = {2021},
  url          = {https://doi.org/10.1609/aaai.v35i6.16640},
  doi          = {10.1609/AAAI.V35I6.16640},
  bibsource    = {dblp computer science bibliography, https://dblp.org}
}

@article{MascleFL23,
  author       = {Corto Mascle and
                  Nathana{\"{e}}l Fijalkow and
                  Guillaume Lagarde},
  title        = {Learning temporal formulas from examples is hard},
  journal      = {CoRR},
  volume       = {abs/2312.16336},
  year         = {2023},
  url          = {https://doi.org/10.48550/arXiv.2312.16336},
  doi          = {10.48550/ARXIV.2312.16336},
  eprinttype   = {arXiv},
  eprint       = {2312.16336},
  bibsource    = {dblp computer science bibliography, https://dblp.org}
}

@article{RahaRFN26,
  author       = {Ritam Raha and
                  Rajarshi Roy and
                  Nathana{\"{e}}l Fijalkow and
                  Daniel Neider},
  title        = {A scalable anytime algorithm for learning fragments of linear temporal
                  logic},
  journal      = {Formal Methods Syst. Des.},
  volume       = {68},
  number       = {1},
  pages        = {2},
  year         = {2026},
  url          = {https://doi.org/10.1007/s10703-025-00489-y},
  doi          = {10.1007/S10703-025-00489-Y},
  bibsource    = {dblp computer science bibliography, https://dblp.org}
}

@inproceedings{
HahnEtAl21,
title={Teaching Temporal Logics to Neural Networks},
author={Christopher Hahn and Frederik Schmitt and Jens U. Kreber and Markus Norman Rabe and Bernd Finkbeiner},
booktitle={International Conference on Learning Representations},
year={2021},
url={https://openreview.net/forum?id=dOcQK-f4byz}
}

@inproceedings{BordaisN025,
  author       = {Benjamin Bordais and
                  Daniel Neider and
                  Rajarshi Roy},
  editor       = {Olaf Beyersdorff and
                  Michal Pilipczuk and
                  Elaine Pimentel and
                  Kim Thang Nguyen},
  title        = {The Complexity of Learning {LTL}, {CTL} and {ATL} Formulas},
  booktitle    = {42nd International Symposium on Theoretical Aspects of Computer Science,
                  {STACS} 2025, Jena, Germany, March 4-7, 2025},
  series       = {LIPIcs},
  pages        = {19:1--19:20},
  publisher    = {Schloss Dagstuhl - Leibniz-Zentrum f{\"{u}}r Informatik},
  year         = {2025},
  url          = {https://doi.org/10.4230/LIPIcs.STACS.2025.19},
  doi          = {10.4230/LIPICS.STACS.2025.19},
  bibsource    = {dblp computer science bibliography, https://dblp.org}
}

@inproceedings{Changjian25,
author = {Zhang, Changjian and Kapoor, Parv and Dardik, Ian and Cui, Leyi and Meira-G\'{o}es, R\^{o}mulo and Garlan, David and Kang, Eunsuk},
title = {Constrained LTL Specification Learning from Examples},
year = {2025},
isbn = {9798331505691},
publisher = {IEEE Press},
url = {https://doi.org/10.1109/ICSE55347.2025.00160},
doi = {10.1109/ICSE55347.2025.00160},
booktitle = {Proceedings of the IEEE/ACM 47th International Conference on Software Engineering},
pages = {629–641},
numpages = {13},
location = {Ottawa, Ontario, Canada},
series = {ICSE '25}
}

@inproceedings{Pnueli77,
  author    = {Amir Pnueli},
  title     = {The Temporal Logic of Programs},
  booktitle = {Proceedings of the 18th Annual Symposium on Foundations of Computer Science (FOCS)},
  pages     = {46--57},
  year      = {1977},
  publisher = {IEEE},
  doi       = {10.1109/SFCS.1977.32}
}

@inproceedings{SchuppanBiere05,
  author    = {Viktor Schuppan and Armin Biere},
  title     = {Shortest Counterexamples for Symbolic Model Checking of {LTL} with Past},
  booktitle = {Tools and Algorithms for the Construction and Analysis of Systems (TACAS)},
  series    = {Lecture Notes in Computer Science},
  volume    = {3440},
  pages     = {493--509},
  year      = {2005},
  publisher = {Springer},
  doi       = {10.1007/978-3-540-31980-1_32}
}

@inproceedings{AmmonsEtAl03,
  author    = {Glenn Ammons and David Mandelin and Rastislav Bod{\'i}k and James R. Larus},
  title     = {Debugging Temporal Specifications with Concept Analysis},
  booktitle = {Proceedings of the ACM SIGPLAN Conference on Programming Language Design and Implementation (PLDI)},
  pages     = {182--195},
  year      = {2003},
  publisher = {ACM},
  doi       = {10.1145/781131.781152}
}

@inproceedings{NeiderGavran18,
  author    = {Daniel Neider and Ivan Gavran},
  title     = {Learning Linear Temporal Properties},
  booktitle = {Formal Methods in Computer-Aided Design (FMCAD)},
  pages     = {1--10},
  year      = {2018},
  publisher = {IEEE},
  doi       = {10.23919/FMCAD.2018.8603016}
}

@inproceedings{CamachoMcIlraith19,
  author    = {Andr{\'e}s Camacho and Sheila A. McIlraith},
  title     = {Learning Interpretable Models Expressed in Linear Temporal Logic},
  booktitle = {Proceedings of the International Conference on Automated Planning and Scheduling (ICAPS)},
  year      = {2019}
}

@article{GavranDarulovaMajumdar20,
  author  = {Ivan Gavran and Eva Darulova and Rupak Majumdar},
  title   = {Interactive Synthesis of Temporal Specifications from Examples and Natural Language},
  journal = {Proceedings of the ACM on Programming Languages},
  volume  = {4},
  number  = {OOPSLA},
  pages   = {201:1--201:26},
  year    = {2020},
  publisher = {ACM},
  doi     = {10.1145/3428269}
}

@article{GoldmanKearns95,
  author  = {Sally A. Goldman and Michael J. Kearns},
  title   = {On the Complexity of Teaching},
  journal = {Journal of Computer and System Sciences},
  volume  = {50},
  number  = {1},
  pages   = {20--31},
  year    = {1995},
  doi     = {10.1006/jcss.1995.1003}
}

@inproceedings{LiEtAlVISSOFT23LTLTimelines,
  author    = {Runming Li and Keerthana Gurushankar and Marijn J. H. Heule and Kristin Y. Rozier},
  title     = {What’s in a Name? {L}inear {T}emporal {L}ogic Literally Represents Time Lines},
  booktitle = {Proceedings of the {IEEE} Working Conference on Software Visualization (VISSOFT)},
  year      = {2023},
  pages     = {73--83},
  doi       = {10.1109/VISSOFT60811.2023.00018}
}

@article{WangGamboaRozierSCP2026WEST,
  author  = {Zili Wang and Laura P. Gamboa Guzman and Kristin Y. Rozier},
  title   = {{WEST}: Interactive Validation of Mission-time Linear Temporal Logic ({MLTL})},
  journal = {Science of Computer Programming},
  volume  = {248},
  pages   = {103365},
  year    = {2026},
  doi     = {10.1016/j.scico.2025.103365}
}

@inproceedings{FortinEtAlKR2022UniqueTemporalInstanceQueries,
  author    = {Marie Fortin and Boris Konev and Vladislav Ryzhikov and Yury Savateev and Frank Wolter and Michael Zakharyaschev},
  title     = {Unique Characterisability and Learnability of Temporal Instance Queries},
  booktitle = {Proceedings of the 19th International Conference on Principles of Knowledge Representation and Reasoning (KR 2022)},
  year      = {2022},
  pages     = {163--173},
  doi       = {10.24963/kr.2022/17},
  url       = {https://proceedings.kr.org/2022/17/}
}

@inproceedings{JungEtAlKR2024UniqueTemporalQueriesOntology,
  author    = {Jean Christoph Jung and Vladislav Ryzhikov and Frank Wolter and Michael Zakharyaschev},
  title     = {Unique Characterisability and Learnability of Temporal Queries Mediated by an Ontology},
  booktitle = {Proceedings of the 21st International Conference on Principles of Knowledge Representation and Reasoning (KR 2024)},
  year      = {2024},
  pages     = {487--497},
  doi       = {10.24963/kr.2024/46},
  url       = {https://proceedings.kr.org/2024/46/}
}

\appendix

\section{Proofs for Section~\ref{sec:prelim}}

\begin{lemma}\label{lem:finite-subword-characterization}
Let $w$ be an ordinal word and let
$\phi \in \LTL_{\ltlFS,\land,\lor,\top,\bot}[\ap]$.
Then
\[
w \models \phi
\iff
\exists u \text{ finite such that } u \models \phi \ \land\ u \le_{\mathrm{hom}} w .
\]
\end{lemma}

\begin{proof}
($\Leftarrow$)
Assume there exists a finite word $u$ such that $u \models \phi$ and
$u \le_{\mathrm{hom}} w$.
By $\le_{\mathrm{hom}}$-monotonicity of
$\LTL_{\ltlFS,\land,\lor,\top,\bot}$ formulas, it follows that
$w \models \phi$.

\smallskip
($\Rightarrow$)
Assume $w \models \phi$.
Since $\phi$ belongs to the fragment
$\LTL_{\ltlFS,\land,\lor,\top,\bot}$, all temporal requirements in $\phi$
are strict-eventualities.
Hence every subformula of the form $\ltlFS(\psi)$ occurring in $\phi$
is witnessed at some finite successor position of $w$.

Let $i_1 < \dots < i_m$ be the minimal witness positions in $w$ for all
$\ltlFS$-subformulas of $\phi$.
Define the finite word
\[
u := w[0] \cdot w[i_1] \cdot \ldots \cdot w[i_m].
\]

By construction, all strict-eventuality requirements of $\phi$ are satisfied in $u$,
and since $\phi$ contains no negation and only monotone Boolean connectives,
we obtain $u \models \phi$.

Define a homomorphism
\[
h(0)=0, \qquad h(j)=i_j \ \text{ for } j>0 .
\]
Then $h$ is order-preserving and label-preserving, hence
$u \le_{\mathrm{hom}} w$.
\end{proof}

\thmFiniteInfiniteEquivalence*

\begin{proof}
The left-to-right direction is immediate.
For the right-to-left direction,
assume that $\semafin{\phi}=\semafin{\psi}$.
Let $w$ be an arbitrary ordinal word and suppose that $w \models \phi$.
By \cref{lem:finite-subword-characterization}, there exists a finite word
$u$ such that $u \models \phi$ and $u \le_{\mathrm{hom}} w$.
It follows that $u \models \psi$.
Applying \cref{lem:finite-subword-characterization} again, it follows that
$w \models \psi$. Similarly, from $w \models \psi$ it follows that $w \models \phi$.
Hence $\sema{\phi}=\sema{\psi}$.

For $\LTL_{\ltlX,\land,\neg}$,  assume that $\semafin{\phi}=\semafin{\psi}$. Let $n$ be a number greater than the maximum nesting depth of temporal operators in $\phi$ and $\psi$. 
Let $w$ be an arbitrary infinite ordinal word and let
$w'$ be the length-$n$ prefix of $w$. 
Then $w\models\phi$ iff $w'\models\phi$ iff 
$w'\models\psi$ iff $w\models\psi$.
Therefore, $\sema{\phi}=\sema{\psi}$.

\end{proof}

\thmRegularExamples*
\begin{proof}
This follows from results in~\cite{BuchiSiefkes1973} and~\cite{DemriRabinovich2010}.
B{\"u}chi and Siefkes~\cite{BuchiSiefkes1973} developed a model of automata on countable ordinals  extending
B{\"u}chi automata for $\omega$-words, and they showed that recognizable languages have regular (or, ``rational'') representatives, i.e., words that have a finite presentation
using concatenation and $\omega$-iteration. As shown 
in \cite[Lemma~2.2 and~2.3]{DemriRabinovich2010}, every LTL-formula $\chi$ can be translated into an equivalent ordinal automaton $A_\chi$. 
It follows that if $\sema{\phi}\setminus\sema{\psi}$ is non-empty, then  $L(A_{\varphi\land\neg\psi})$ is non-empty and hence contains a regular transfinite word.
\end{proof}

\section{Proofs for Section~\ref{sec:finite}}

\thmMainFinite*

We prove the positive cases and the complementary negative cases separately. 
As illustrated in Figure~\ref{fig:finite-fragments-diagram}, every fragment covered by Theorem~\ref{thm:ltlfin} is either subsumed by one of the positive fragments of \autoref{thm:finite-char-various-fragments} or subsumes one of the negative fragments of \autoref{thm:no-finite-char-various-fragments}.

\begin{theorem}[Positive results]\label{thm:finite-char-various-fragments} 
Let $\ap$ be any finite non-empty set.   
\begin{enumerate}
        
    \item \label{B11} $\LTL_{\ltlFS,\ltlX,\land,\top}[\ap]$ admits finite characterizations over $\finwords$.
        
    \item \label{B12}$\LTL_{\ltlF,\lor,\top,\bot}[\ap]$  admits finite characterizations over $\finwords$.
        
    \item \label{B13} $\LTL_{\ltlF,\neg,\top,\bot}[\ap]$  admits finite characterizations over $\finwords$.

    \item \label{B14} $\LTL_{\ltlF,\ltlFS,\ltlX,\top}[\ap]$  admits finite characterizations over $\finwords$. 

    \item \label{B15} $\LTL_{\ltlX,\neg}[\ap]$  admits finite characterizations over $\finwords$. 

    \item \label{B16} $\LTL_{\ltlFS,\neg}[\ap]$  admits finite characterizations over $\finwords$. 
    \end{enumerate}
\end{theorem}

\begin{proof}
1.  This is due to \cite{FortinEtAlKR2022UniqueTemporalInstanceQueries}. For the sake of being self-contained, we include the argument.
    Let a $\LTL_{\ltlFS,\ltlX,\land,\top}$-formula $\phi$ be given, and 
    let $n$ be the maximal nesting depth of temporal operators in $\phi$.
    Consider the finite word $w = (\ap)^{n+1}$. It can be shown by a straightforward induction that, for all $\LTL_{\ltlFS,\ltlX,\land,\top}$-formulas $\psi$, $w\models\phi$ if and only if the maximal nesting depth of temporal operators in $\phi$ is at most $n$. Since there are only finitely many non-equivalent formulas with maximal temporal nesting depth $n$, we can 
    uniquely characterize $\phi$ by including $w$ as a positive example, 
    together with, for each formula $\psi$ of maximal temporal nesting depth $n$
    that is not equivalent to $\phi$, a distinguishing example.

2. There are only finitely many formulas up to equivalence in this fragment (in a given finite set $\ap$). This follows immediately from the fact that $\ltlF(\phi\lor\psi)\equiv \ltlF(\phi)\lor F(\psi)$ and $\ltlF\ltlF\phi\equiv \ltlF\phi$ and $\ltlF\top\equiv \top$ and $\ltlF\bot\equiv \bot$. 

3.  Recall that $\ltlG\phi$ is shorthand for $\neg\ltlF\neg\phi$.
    By rewriting to negation normal form and applying
    the equivalences $\ltlF\bot\equiv\bot$, $\ltlG\bot\equiv\bot$, $\ltlF\top\equiv\top$
    and $\ltlG\top\equiv\top$, every formula in this fragment can be
    written in the form
    \begin{enumerate}
        \item $o_1\ldots, o_n p$ or 
        \item $o_1\ldots, o_n \neg p$ or 
        \item $\top$ or 
        \item $\bot$
    \end{enumerate}
    with $o_i\in\{\ltlF,\ltlG\}$. 
    Furthermore, we have that 
    $\ltlF\ltlF p\equiv \ltlF p$ and $\ltlG\ltlG p\equiv \ltlG p$, and 
    $\ltlF\ltlG\ltlF p\equiv \ltlG \ltlF p$ and
    $\ltlG\ltlF\ltlG p\equiv \ltlF\ltlG p$. Thus, only finitely many possible formulas up to equivalence, and therefore each formula trivially admits a finite characterization.

4.    By applying the equivalences $\ltlF\ltlFS p\equiv \ltlFS\ltlF p\equiv \ltlFS p$ and $\ltlX\ltlF p\equiv \ltlF\ltlX p=\ltlFS p$ we can rewrite every
    formula in this fragment to a formula the form 
    \begin{enumerate}
        \item $\ltlX^n\ltlF p$ or 
        \item $\ltlX^n\ltlF \top$ or
        \item $\ltlX^n p$ or 
        \item $\ltlX^n\top$
    \end{enumerate}
    Let us refer to the number $n$ as the \emph{depth} of the formula in question.
    Let $\phi$ be any formula of one of the above forms, with depth $n$.  
    Let $w=(\ap)^{n+1}$. Then $w\models\phi$ and, for all formulas $\psi$ of the above forms of depth greater than $n$, we have that 
    $w\not\models\psi$. Thus, we can construct a unique characterization for $\phi$
    by taking $w$ as a positive example and including, for each formula $\psi$
    of the above forms of depth at most $n$ that is not equivalent to $\phi$, a distinguishing example.

5. Let $\ltlX^0\phi$ be short for $\neg \ltlX \neg\phi$,
which means ``$\phi$ is true at the next timepoint or there is no next timepoint''. 
Then every formula in this fragment can be rewritten to the form
\[ o_1\ldots o_n \alpha\]
with $o_i\in\{\ltlX,\ltlX^0\}$ ($n\geq 0$)
and $\alpha$ of the form $p$ or $\neg p$. Such a formula is uniquely characterized by the following set of labeled examples:
\begin{itemize}
    \item $\emptyset^{k}$ for $k=0\ldots n-1$ labeled as positive/negative according to whether $o_{k+1}=\ltlX^0$ or $o_{k+1}=\ltlX$;
    \item $\emptyset^{n} \{p\}$ and $\emptyset^{n+1}$ labeled as positive/negative according to whether $\alpha=p$ or $\alpha=\neg p$.
\end{itemize}

6. Let $\hat{\ltlG}$ be short for $\neg\ltlFS\neg$. By rewriting to negation normal form and applying the equivalences $\hat{\ltlG}\ltlFS\phi\equiv\ltlG\bot$ and $\ltlFS\hat\ltlG\phi\equiv \ltlFS\top$  and $\hat{\ltlG}\top\equiv\top$ and $\ltlFS\bot\equiv\bot$, we can rewrite all formulas in this fragment to one of the following forms:
\begin{enumerate}
    \item $\ltlFS^n \alpha$
($n> 0$) with $\alpha$ of the form $p$ or $\neg p$ or $\top$, or
\item  $\hat{\ltlG}^n \beta$ ($n>0$) with $\beta$ of the form $p$ or $\neg p$ or 
$\bot$.
    \item $p$ or $\neg p$
\end{enumerate}
A formula of the first type is uniquely characterized by the examples $\emptyset^{n-1}$ and $\ap^{n-1}$ labeled as a negative example, together with
    the examples $\emptyset^n\{p\}$ and $\emptyset^{n+1}$ labeled according to whether $\alpha$ is of the form $p$ or $\neg p$ or $\top$.
    A formula of the second type is uniquely characterized by the examples $\emptyset^{n-1}$ and $\ap^{n-1}$ as a positive examples, and the examples
    $\emptyset^{n}$ and $\emptyset^{n-1}\{p\}$,
    labeled according to whether $\alpha$ is of the form $p$ or $\neg p$ or $\bot$.
    Finally, a formula of the third type is uniquely characterized by $\{p\}$ and $\emptyset$ labeled according to whether 
    the formula is $p$ or $\neg p$.
\end{proof}

\begin{theorem}[Negative Results]\label{thm:no-finite-char-various-fragments}
    Let $|\ap|\geq 2$.
    \begin{enumerate}
        \item \label{B21} $\LTL_{\ltlF,\land}[\mathit{\ap}]$ does not admit finite characterizations over $\finwords$ \cite[Example 1]{FortinEtAlKR2022UniqueTemporalInstanceQueries}
          
        \item \label{B22} $\LTL_{\ltlX,\lor}[\mathit{\ap}]$
        and $\LTL_{\ltlFS,\lor,}[\mathit{\ap}]$ do not  admit finite characterizations over $\finwords$.
        
        \item \label{B23} $\LTL_{\ltlX,\bot}[\mathit{\ap}]$
        and $\LTL_{\ltlFS,\bot}[\mathit{\ap}]$ do not  admit finite characterizations over $\finwords$.
        
        \item \label{B24} $\LTL_{\ltlX,\land,\neg}[\mathit{\ap}]$, $\LTL_{\ltlF,\land,\neg}[\mathit{\ap}]$
        and $\LTL_{\ltlFS,\land,\neg}[\ap]$
        do not admit finite characterizations over $\finwords$.

        \item \label{B25} $\LTL_{\ltlU}[\ap]$  does not admit finite unique characterizations over $\finwords$.

    \item \label{B26} $\LTL_{\ltlFS,\ltlX,\neg}[\ap]$ does not  admit finite characterizations over $\finwords$.

    \item \label{B27} $\LTL_{\ltlF,\ltlFS,\neg}[\ap]$ does not 
    admit finite characterizations over $\finwords$.
    \end{enumerate}
\end{theorem}

\begin{proof}
    \begin{enumerate}
        \item As observed in \cite{FortinEtAlKR2022UniqueTemporalInstanceQueries}, the formula $\ltlF(p\land q)$ cannot be distinguished from all formulas of the form $\ltlF(p\land\ltlF(q\land\ltlF(p\land\cdots)))$ using only finitely many labeled examples. 
        
       \item The formula $p$ cannot be distinguished from all formulas of the form $p\lor\ltlX^n p$, respectively,
       $p\lor\ltlFS^n p$ ($n>0$) 
       using only finitely many labeled examples. 
        
    \item 
        The argument is similar to the one above: we 
        cannot distinguish $\bot$ from all formulas of the form $\ltlX^n p$, respectively,
        $\ltlFS^n p$, using only finitely many labeled examples.

    \item 
        For $\LTL_{\ltlX,\land,\neg}[\mathit{\ap}]$ and $\LTL_{\ltlFS,\land,\neg}[\ap]$, the result follows from item 2 above (since disjunction is definable from conjunction and negation). 
        For $\LTL_{\ltlF,\land,\neg}[\mathit{\ap}]$, we reason as follows: we can express
        $\bot$ as $p\land\neg p$. We cannot distinguish
        $\bot$ from all formulas of the form $p\land \ltlF(\neg p\land \ltlF(p\land \ltlF(\neg p\ldots)))$ using only finitely many examples. 

    \item 
Let \(\phi := (p \ltlU q)\ltlU r\), and define \(\psi_0 := r\) and \(\psi_{n+1} := p \ltlU (q \ltlU \psi_n)\). Then \(\phi\) allows arbitrarily many alternations of \(p\) and \(q\) before the first occurrence of \(r\), whereas \(\psi_n\) bounds this number by \(n\). In particular, \((\{p\}\{q\})^{n+1}\{r\} \models \phi\) but \(\not\models \psi_n\).

Let \(S\) be any finite sample. For each positive example \(w \in S\) with \(w \models \phi\), fix a finite witness position of \(r\), and let \(m\) be the maximal number of \(p/q\)-alternations before such a witness over all these examples. Choose \(n>m\). Then \(\phi\) and \(\psi_n\) agree on all examples in \(S\), although they are not equivalent. Hence no finite sample uniquely characterizes \(\phi\).
        \item Consider the formula
$\neg\ltlFS\neg\ltlFS p$, it is equivalent to $\neg\ltlX\top$. We cannot distinguish this formula from $\neg\ltlX\neg\ltlX^n p$, for all $n$, using only finitely many finite examples: an example that distinguishes the two formulas needs to be a word of length at least $n$.

        \item Recall that $\ltlG\phi$ is short for
        $\neg\ltlF\neg\phi$. Since 
        $\ltlG\ltlFS\phi$ is equivalent, over finite words, to $\bot$, and we
        already know from item 3 above that
        $\LTL_{\ltlFS,\bot}[\ap]$ does not admit finite characterizations over $\finwords$,
        it follows
        that $\LTL_{\ltlF,\ltlFS,\neg}[\ap]$ does not admit finite characterizations over $\finwords$ either. 
    \end{enumerate}
\end{proof}

\def\labeldist{1.0mm}
\definecolor{darkgreen}{rgb}{0,0.5,0}

\begin{figure}
\scalebox{0.8}{
    \centering
    \label{fig:nfa}
\begin{tikzpicture}[
    >={Stealth},
    node distance=1.4cm,
    start/.style={circle, fill=black, minimum size=4pt, inner sep=0pt, initial, initial where=above, initial text=},
    vertex/.style={rectangle, rounded corners=1.4mm,  draw, minimum size=6mm, inner sep=1pt, font=\scriptsize},
    redvertex/.style={vertex, fill=red!30},
    greenvertex/.style={vertex, fill=green!30},
    label/.style={font=\scriptsize\bfseries}
]

\node[start] (s) {};
\node[left=8mm of s] (start) {};

\node[greenvertex, below=of s] (fs) {$\ltlFS$};
\node[greenvertex, right=of fs] (fsand) {$\ltlFS, \land$};
\node[greenvertex, right=of fsand] (fsxandtop) {$\ltlFS, \ltlX, \land, \top$};
\node[label,below=\labeldist of fsxandtop] (fsxandtopl) {\textcolor{darkgreen}{\ref{thm:finite-char-various-fragments}.\ref{B11}}};
\node[greenvertex, right=of fsxandtop] (xand) {$\ltlX, \land$};
\node[redvertex, right=4.0cm of s] (u) {$\ltlU$};
\node[label,right=\labeldist of u] (ul) {\textcolor{red}{\ref{thm:no-finite-char-various-fragments}.\ref{B25}}};
\node[greenvertex, right=of xand] (x) {$\ltlX$};
\node[greenvertex, below=of fs] (fsnot) {$\ltlFS, \neg$};
\node[label,below=\labeldist of fsnot] (fsnotl) {\textcolor{darkgreen}{\ref{thm:finite-char-various-fragments}.\ref{B16}}};
\node[greenvertex, below=of x] (xnot) {$\ltlX, \neg$};
\node[label,below=\labeldist of xnot] (xnotl) {\textcolor{darkgreen}{\ref{thm:finite-char-various-fragments}.\ref{B15}}};
\node[greenvertex, left=of fs] (fnot) {$\ltlF, \neg$};
\node[greenvertex, left=of fnot] (f) {$\ltlF$};
\node[greenvertex, below=of fnot] (fnottopbot) {$\ltlF, \neg, \top, \bot$};
\node[label,below=\labeldist of fnottopbot] (fnottopbotl) {\textcolor{darkgreen}{\ref{thm:finite-char-various-fragments}.\ref{B13}}};
\node[greenvertex, left=of f] (for) {$\ltlF, \lor$};
\node[greenvertex, below=of for] (fortopbot) {$\ltlF, \lor, \top, \bot$};
\node[label,below=\labeldist of fortopbot] (fortopbotl) {\textcolor{darkgreen}{\ref{thm:finite-char-various-fragments}.\ref{B12}}};
\node[circle, fill=black, minimum size=4pt, inner sep=0pt, below=of f] (belowf) {};
\node[greenvertex, below=of belowf] (ffsxtop) {$\ltlF, \ltlFS,\ltlX, \top$};
\node[label,below=\labeldist of ffsxtop] (ffsxtopl) {\textcolor{darkgreen}{\ref{thm:finite-char-various-fragments}.\ref{B14}}};
\node[greenvertex, right=of ffsxtop] (fx) {$\ltlF, \ltlX$};
\node[greenvertex, left=of ffsxtop] (ffs) {$\ltlF, \ltlFS$};
\node[circle, fill=black, minimum size=4pt, inner sep=0pt, below=of ffs] (belowffs) {};
\node[redvertex, below=of belowffs] (fand) {$\ltlF, \land$};
\node[label,below=\labeldist of fand] (fandl) {\textcolor{red}{\ref{thm:no-finite-char-various-fragments}.\ref{B21}}};
\node[redvertex, right=of fand] (ffsxnot) {$\ltlF, \ltlFS, \neg$};
\node[label,below=\labeldist of ffsxnot] (ffsxnotl) {\textcolor{red}{\ref{thm:no-finite-char-various-fragments}.\ref{B27}}};
\node[redvertex, right=of ffsxnot] (fsxorbot) {$\ltlFS / \ltlX, \lor / \bot$};
\node[label,below=\labeldist of fsxorbot] (fsxorbotl) {\textcolor{red}{\ref{thm:no-finite-char-various-fragments}.\ref{B22}, \ref{thm:no-finite-char-various-fragments}.\ref{B23}}};
\node[circle, fill=black, minimum size=4pt, inner sep=0pt, below=of fx] (belowfx) {};
\node[circle, fill=black, minimum size=4pt, inner sep=0pt, above=of ffsxnot] (aboveffsxnot) {};
\node[greenvertex, right=of fsnot] (fsx) {$\ltlFS, \ltlX$};
\node[circle, fill=white, minimum size=4pt, inner sep=0pt, below=of fsx] (belowfsx) {};
\node[circle, fill=white, minimum size=4pt, inner sep=0pt, below=of belowfsx] (belowbelowfsx) {};
\node[redvertex, below=of belowbelowfsx] (fsxnot) {$\ltlFS, \ltlX, \neg$};
\node[label,below=\labeldist of fsxnot] (fsxnotl) {\textcolor{red}{\ref{thm:no-finite-char-various-fragments}.\ref{B26}}};
\node[redvertex, right=of fsxnot] (fsxorbot2) {$\ltlFS / \ltlX, \lor / \bot$};
\node[label,below=\labeldist of fsxorbot2] (fsxorbot2l) {\textcolor{red}{\ref{thm:no-finite-char-various-fragments}.\ref{B22}, \ref{thm:no-finite-char-various-fragments}.\ref{B23}}};
\node[circle, fill=white, minimum size=4pt, inner sep=0pt, above=of fsxorbot2] (abovefsxorbot2) {};
\node[circle, fill=black, minimum size=4pt, inner sep=0pt, above=of abovefsxorbot2] (aboveabovefsxorbot2) {};
\node[redvertex, above=of f] (fand2) {$\ltlF, \land$};
\draw[->] (s) -- node[midway, above] {$\ltlU$} (u);
\draw[->] (s) -- node[midway, right] {$\ltlFS$} (fs);
\draw[->] (s) -- node[midway, above] {$\ltlF$} (f);
\draw[->] (s) -- node[midway, above] {$\ltlX$} (x);
\draw[->] (fs) -- node[midway, above] {$\land$} (fsand);
\draw[->] (x) -- node[midway, above] {$\land$} (xand);
\draw[->] (fs) -- node[midway, left] {$\neg$} (fsnot);
\draw[->] (x) -- node[midway, right] {$\neg$} (xnot);
\draw[->] (xand) -- node[midway, above] {$\top$} (fsxandtop);
\draw[->] (fsand) -- node[pos=0.5, above, yshift=-2pt] {$\ltlX, \top$} (fsxandtop);
\draw[->] (f) -- node[midway, above] {$\neg$} (fnot);
\draw[->] (f) -- node[midway, above] {$\lor$} (for);
\draw[->] (fnot) -- node[midway, left] {$\top, \bot$} (fnottopbot);
\draw[->] (for) -- node[midway, right] {$\top, \bot$} (fortopbot);
\draw (f) -- (belowf);
\draw[->] (ffs) -- node[pos=0.5, above, yshift=-2pt] {$\ltlX, \top$} (ffsxtop);
\draw[->] (fx) -- node[midway, above] {$\top$} (ffsxtop);
\draw[->] (belowf) -- node[midway, above] {$\ltlFS$} (ffs);
\draw[->] (belowf) -- node[midway, above] {$\ltlX$} (fx);
\draw (ffs) -- node[midway, left] {} (belowffs);
\draw (fx) to node[midway, above] {} (belowffs);
\draw[ bend left=4] (belowf) to node[midway, left] {} (belowffs);
\draw[->, bend right=25] (for) to node[midway, left] {$\neg$} (fand);
\draw (ffs) to node[midway, left] {} (belowfx);
\draw (fx) to node[midway, left] {} (belowfx);
\draw (ffsxtop) to node[midway, left] {} (belowfx);
\draw[->] (belowffs) to node[midway, left] {$\land$} (fand);
\draw[->] (belowfx) to node[midway, right] {$\lor, \bot$} (fsxorbot);
\draw (ffs) to node[midway, right] {} (aboveffsxnot);
\draw (fx) to node[midway, right] {} (aboveffsxnot);
\draw[->] (aboveffsxnot) to node[midway, right] {$\neg$} (ffsxnot);
\draw[->, bend left=25] (fs) to node[midway, above] {$\top$} (fsxandtop);
\draw[->, bend right=25] (x) to node[midway, above] {$\top$} (fsxandtop);
\draw[->] (fs) to node[midway, above] {$\ltlX$} (fsx);
\draw[->] (fsx) to node[midway, right] {$\neg$} (fsxnot);
\draw (fsnot) to node[near end, above] {$\top$} (aboveabovefsxorbot2);
\draw (xnot) to node[near end, above] {$\top$} (aboveabovefsxorbot2);
\draw (fsxandtop) to node[near end, above] {} (aboveabovefsxorbot2);
\draw (fsand) to node[near end, above] {$\neg$} (aboveabovefsxorbot2);
\draw (xand) to node[near end, above] {$\neg$} (aboveabovefsxorbot2);
\draw[->] (aboveabovefsxorbot2) to node[midway, right] {$\lor, \bot$} (fsxorbot2);
\draw[bend left=10] (fs)  to node[near end, above] {} (aboveabovefsxorbot2);
\draw[bend right=10] (x)  to node[near end, above] {} (aboveabovefsxorbot2);
\draw[->] (f) to node[midway, right] {$\land$} (fand2);
\node[label,left=\labeldist of fand2] (fand2l) {\textcolor{red}{\ref{thm:no-finite-char-various-fragments}.\ref{B21}}};

\end{tikzpicture}
}
    \caption{This figure illustrates the result of \autoref{thm:ltlfin}. It can be viewed as an automaton over the alphabet $\{\ltlU,\ltlF,\ltlFS,\ltlX,\land,\lor,\neg,\top,\bot\}$ that accepts exactly those words that satisfy the following conditions:
(1) they include at least one temporal operator;
(2) occurrences of operators respect the order induced by the alphabet; and
(3) the set of operators is not a superset of any of the negative fragments listed in \autoref{thm:no-finite-char-various-fragments}.
The leaves of the figure correspond to negative fragments (shown in red), each annotated with the corresponding item of \autoref{thm:no-finite-char-various-fragments} (or to nodes from which no further operators can be read), whereas inner nodes correspond to positive fragments (shown in green).
It follows that the accepted words correspond to sets of operators that are subsets of the positive fragments listed in \autoref{thm:finite-char-various-fragments}. To reduce clutter, some negative fragments appear more than once.}
    \label{fig:finite-fragments-diagram}
\end{figure}

\section{Proofs for Section~\ref{sec:omega}}

\thmXpos*

\begin{proof}
   Let $\phi$ be a formula in the fragment, and let $n$ be the maximal nesting depth of $\ltlX$-operators in the formula. Clearly, 
   the truth of $\phi$ in an example depends only on the first $n$ positions. We apply
   a brute force construction: 
   \begin{enumerate}
       \item We include, in our set of labeled examples, all words of length at most $n$ (labeled according to $\phi$)
       \item In addition, for each word of length $n$, we include the $\omega$-words
       $w \cdot \emptyset^\omega$ and 
       $w \cdot \ap^\omega$ (labeled according to $\phi$).
   \end{enumerate}
   We claim that the resulting set of labeled examples uniquely characterizes $\phi$ with respect to $\LTL_{\ltlX,\land,\lor,\top,\bot}[\ap]$ over $\allwords$. Indeed, if a formula $\psi\in \LTL_{\ltlX,\land,\lor,\top,\bot}[\ap]$
   fits these examples, then it must agree with $\phi$ on all examples. This is clearly so for words of length at most $n$. For words $w$ longer than $n$, let $w'$ and $w''$ be copies of $w$ in which all positions greater than $n$ are assigned $\emptyset$, respectively, $\ap$. Since
   the truth of $\phi$ depends only on the length-$n$ prefix of the word, there are two
   possibilities: either $w,w',w''$ all satisfy $\phi$,
   or none of them satisfies $\phi$. In the former case, it follows that $w'$ and $w''$ both satisfy $\psi$, while in the latter case, it follows that $w'$ and $w''$ both falsify $\psi$ (since these examples were included in the set). It follows by monotonicity that, in the former case, $w$ satisfies $\psi$, and, in the latter case, $w$ falsifies $\psi$. 

   The proof for $\LTL_{\ltlX,\land,\lor,\top,\bot}[\ap]$ over $\words^{=\omega}$ is similar except that the labeled examples of type 1 are omitted.
\end{proof}

\thmomegacounterexample*

\begin{proof}
Let 
\[\begin{array}{lll}
\phi &=& \ltlF(p\land q\land \ltlF (r\land\ltlF(p\land q))) \\
\phi_k &=&
\ltlF(\underbrace{p\land\ltlF q\land \ltlF(p\land \ltlF(q\land\cdots}_{\text{$k$ alternations of $p$ and $q$}}\land\ltlF (r\land\ltlF(p\land q)))))
\end{array}\]
Clearly, $\sema{\phi}^{\leq \omega}\subseteq\sema{\phi_k}^{\leq \omega}$
and 
$\sema{\phi_k}^{= \omega}\not\subseteq\sema{\phi}^{= \omega}$,
for all $k\in\mathbb{N}$.
Let $E$ be a finite set of labeled examples consisting of words of length at most $\omega$, and suppose that $\phi$ fits $E$. We will 
show that some $\phi_k$ also fits $E$.
Consider any negative example $(w,-)\in E$. 
Since $w\not\models\phi$, clearly, $w$ cannot contain contain the letter $\{p,q,r\}$.
Similarly, $w$ 
cannot contain infinitely many occurrences both of $\{p,q\}$ and of $\{r\}$.
Thus, there must be some $\ell\in\mathbb{N}$ such that
$w[\ell,\ldots]$ either does not contain $\{p,q\}$ or does not contain $\{r\}$. 
It follows
that, for all $k>\ell$, $w\not\models\phi_k$ 
and hence $\phi_k$ fits $(w,-)$. This is because the sequence of $k$ 
alternations of $p$ and $q$ described in $\phi_k$ must be satisfied within length-$\ell$ prefix of $w$ and $\ell<k$.
Thus, 
if we take
$\ell_{\max}$ to be the maximum value of $\ell$ over all negative examples in $E$, then $\phi_{\ell_{\max}+1}$ fits all negative examples in
$E$. Since $\sema{\phi}^{\leq \omega}\subseteq\sema{\phi_k}^{\leq \omega}$, 
 $\phi_{\ell_{\max}+1}$ fits all positive examples in $E$ as well.
\end{proof}

\section{Proofs for Section~\ref{sec:transfinite}}

We start with a basic monotonicity property of
$\LTL_{\ltlF,\land,\lor,\top,\bot}[\ap]$
under the homomorphic preorder $\leq_{\hom}$.
It states that satisfaction is preserved when a word is extended by a homomorphic embedding.
This will be central for the construction of canonical positive examples and maximal negative examples.

\monotone*

\begin{proof}
Fix $\varphi$ and ordinal words $w,w'$ with $w\leq_{\mathrm{hom}} w'$,
witnessed by a homomorphism $h$.
Recall that $h$ is nondecreasing and continuous, satisfies $h(0)=0$,
and preserves labels in the sense that for all positions $i$, $w[i] \subseteq w'[h(i)]$.
(In particular, $i>0$ implies $h(i) > 0$.)

We prove the statement by structural induction on $\varphi$.

\medskip\noindent\emph{Base cases.}

\smallskip\noindent
($\top$) Trivial, since every word satisfies $\top$.

\smallskip\noindent
($\bot$) Impossible premise: $w\models\bot$ never holds, so the implication is vacuously true.

\smallskip\noindent
($p\in\ap$) If $w\models p$, then by semantics $p\in w[0]$.
Since $h(0)=0$ and $w[0]\subseteq w'[h(0)]=w'[0]$, we get $p\in w'[0]$,
hence $w'\models p$.

\medskip\noindent\emph{Inductive steps.}

Assume the claim holds for $\psi_1,\psi_2$ (and $\psi$).

\smallskip\noindent
($\land$) Let $\varphi=\psi_1\land\psi_2$ and suppose $w\models\varphi$.
Then $w\models\psi_1$ and $w\models\psi_2$.
By the induction hypothesis (IH), $w'\models\psi_1$ and $w'\models\psi_2$.
Hence $w'\models\psi_1\land\psi_2$.

\smallskip\noindent
($\lor$) Let $\varphi=\psi_1\lor\psi_2$ and suppose $w\models\varphi$.
Then $w\models\psi_1$ or $w\models\psi_2$.

If $w\models\psi_i$ for some $i\in\{1,2\}$, the IH yields $w'\models\psi_i$,
and thus $w'\models\psi_1\lor\psi_2$.

\smallskip\noindent
($\ltlFS$) Let $\varphi=\ltlFS\psi$ and suppose $w\models\ltlFS\psi$.
By the strict-$\ltlF$ semantics, there is some position
$i>0$ such that the suffix $w[i..]$ satisfies $\psi$, i.e., $w[i..] \models \psi$.

We next show that $w[i..] \leq_{\mathrm{hom}} w'[h(i)..]$.\\
Define a map $h_i(\alpha) := h(i+\alpha) - h(i)$
(on the ordinal domain of $w[i..]$, reindexed to start at $0$).
By the properties of $h$:
\begin{itemize}
  \item $h_i$ is nondecreasing and continuous (as $h$ is).
  \item $h_i(0) = h(i) - h(i) = 0$.
  \item For all $\alpha$, $w[i..][\alpha] = w[i+\alpha] 
    \subseteq w'[h(i+\alpha)] = w'[h(i)..][h(i+a) - h(i)] = w[h(i)..][h_i(\alpha)]$.
\end{itemize}
Thus $h_i$ witnesses $w[i..] \leq_{\mathrm{hom}} w'[h(i)..]$.

By the IH applied to $\psi$, from $w[i..] \models\psi$ and
$w[i..] \leq_{\mathrm{hom}} w'[h(i)..]$ we obtain $w'[h(i)..] \models \psi$.
Since $i>0$ and $h(0)=0$ with $h$ increasing, we have $h(i)>0$,
so there is a position $j>0$ (namely $j=h(i)$) such that
$w'[j..]\models\psi$. By the semantics of $\ltlFS$,
this exactly means $w'\models\ltlFS\psi$.

\medskip
This covers all constructors of $\mathrm{LTL}_{\ltlFS,\land,\lor,\top,\bot}[\ap]$,
so the induction is complete.
\end{proof}

\subsection{Construction of canonical set}
The canonical set is intended to capture the $\leq_{\hom}$-minimal models of a formula. Its construction follows the syntax of the formula, so that each operator combines canonical witnesses of its subformulas into canonical witnesses for the whole formula. Intuitively, every word in the canonical set satisfies the formula, while every satisfying word contains, via $\leq_{\hom}$, some canonical example as a minimal witness. Hence the language of the formula is exactly the upward closure of its canonical set.

Most cases follow directly from the semantics of the operators. The main subtlety is conjunction: a witness must satisfy both conjuncts simultaneously. For this, given witnesses $w_1$ and $w_2$ for the two subformulas, we construct words that realize all possible compatible interleavings of $w_1$ and $w_2$, preserving order and labels. The first positions of the two witnesses are merged, since both conjuncts are evaluated at the same initial point. After that, the remaining obligations of the two witnesses may occur in any relative order, provided that the internal order of each witness is preserved. When obligations from both witnesses are realized at the same position, their labels are combined by union. Thus, the merged interleaving represents the two witnesses progressing in parallel, rather than forcing one witness to occur entirely before the other. This ensures that all minimal joint witnesses are represented.

\begin{definition}[Merged Interleaving]
Let 
  $u = u_0\cdots u_{m-1}$ 
and 
  $v = v_0 \cdots v_{n-1}$
be nonempty finite words over $\mathcal P(\ap)$.
A finite word 
  $e = e_0\cdots e_{\ell-1}$
is a \emph{merged interleaving} of $u$ and $v$, written $e \in u \mergeop v$,
if there exist strictly increasing sequences of tail indices
\[
1 \le a_1 < \cdots < a_{m-1} \le \ell-1,
\qquad
1 \le b_1 < \cdots < b_{n-1} \le \ell-1,
\]
such that:
\begin{enumerate}
  \item $e_0 = u_0 \cup v_0$;

  \item for each $i=1,\ldots,m-1$,
        \[
        e_{a_i} =
        \begin{cases}
           u_i, & a_i \notin \{b_1,\ldots,b_{n-1}\},\\[2pt]
           u_i \cup v_j, & a_i = b_j\ \text{for a unique } j;
        \end{cases}
        \]

  \item for each $j=1,\ldots,n-1$,
        \[
        e_{b_j} =
        \begin{cases}
           v_j, & b_j \notin \{a_1,\ldots,a_{m-1}\},\\[2pt]
           u_i \cup v_j, & b_j = a_i\ \text{for a unique } i;
        \end{cases}
        \]

  \item every $u_i$ ($i\ge 1$) and every $v_j$ ($j\ge 1$)
        appears exactly once in $e$, either alone or as part of a single merged letter.
\end{enumerate}
For two sets of finite words $U, V$ define: $U \mergeop V = \set{u \mergeop v | u\in U, v\in V}$
\end{definition}

\begin{definition}[Canonical Examples]
The set of canonical examples $E_{\mathrm{can}}(\varphi)$ for a formula $\varphi \in \LTL_{\ltlFS,\land,\lor,\top,\bot}$ is defined by structural induction on~$\varphi$ as follows:
\[
\begin{aligned}
E_{\mathrm{can}}(\bot) &:= \set{}, \\
E_{\mathrm{can}}(\top) &:= \{\;\emptyset \;\}, \\
E_{\mathrm{can}}(p) &:= \{\; \{p\} \;\} \qquad (p \in \ap), \\
E_{\mathrm{can}}(\ltlFS \varphi) &:= \set{\emptyset} \cdot E_{\mathrm{can}}(\varphi), \\
E_{\mathrm{can}}(\varphi \lor \psi) &:= \Ecan(\varphi) \cup \Ecan(\psi),\\
E_{\mathrm{can}}(\varphi \land \psi) &:= E_{\mathrm{can}}(\varphi) \mergeop E_{\mathrm{can}}(\psi).
\end{aligned}
\]
\end{definition}

We next show that the canonical examples are indeed sound,
in the sense that every constructed example satisfies the formula.
This follows by a straightforward structural induction.

\begin{lemma}[Monotone Extension]\label{lem:mon-ext} Let $u,w$ be ordinal words over $2^{\ap}$ such that $w$ is obtained from $u$ by: \begin{enumerate} \item[\emph{(i)}] inserting extra positions \emph{after} the first letter (no reordering or deletion), and \item[\emph{(ii)}] possibly enlarging some letters, i.e.\ for every original position $\alpha$ of $u$, the corresponding position of $w$ carries a superset of labels. \end{enumerate} Then $u \leq_{hom} w$ \end{lemma}

\begin{proof}
Let $\dom(u)=\gamma$ and $\dom(w)=\delta$.

By condition \emph{(i)}, there exists a strictly increasing function
\( h:\gamma \to \delta \) such that \( h(0)=0 \) and for every
\( \alpha<\gamma \), the position \( h(\alpha) \) in \( w \) corresponds to
the original position \( \alpha \) of \( u \).
By condition \emph{(ii)}, for all \( \alpha<\gamma \) we have $u(\alpha)\subseteq w(h(\alpha))$.

Thus \( h \) is a homomorphism from \( u \) to \( w \) according to the
definition of \( \leq_{\mathrm{hom}} \), and therefore
\( u \leq_{\mathrm{hom}} w \).
\end{proof}

\begin{lemma}[Soundness of canonical examples]
\label{lem:can-sound}
For every $\varphi\in \mathrm{LTL}_{\ltlFS,\land,\lor,\top,\bot}[\ap]$
and every $e\in E_{\mathrm{can}}(\varphi)$ we have $e\models\varphi$.
\end{lemma}

\begin{proof}
By structural induction on $\varphi$.

\emph{Base cases.}
If $\varphi=\bot$, then $E_{\mathrm{can}}(\bot)=\emptyset$ and the claim is vacuous.

If $\varphi=\top$, then $E_{\mathrm{can}}(\top)=\{\emptyset\}$ and trivially $\emptyset\models\top$.

If $\varphi=p\in\ap$, then $E_{\mathrm{can}}(p)=\{\{p\}\}$ and clearly $\{p\}\models p$.

\smallskip
\emph{Inductive cases.}

(1) $\varphi=\ltlFS\psi$.

Then $E_{\mathrm{can}}(\ltlFS\psi)=\set{\emptyset}\cdot E_{\mathrm{can}}(\psi)$.

Take $u\in E_{\mathrm{can}}(\psi)$; by IH, $u\models\psi$.
Let $e:=\emptyset\cdot u$,
then $e\models \ltlFS\psi$, witnessed at position $1>0$: $e[1..] = u \models\psi$.

(2) $\varphi=\psi_1\lor\psi_2$.

Here $E_{\mathrm{can}}(\psi_1\lor\psi_2)=
E_{\mathrm{can}}(\psi_1)\cup E_{\mathrm{can}}(\psi_2)$.

Take $e$ in the union; w.l.o.g.\ $e\in E_{\mathrm{can}}(\psi_1)$.
By IH, $e\models\psi_1$, hence $e\models\psi_1\lor\psi_2$.

(3) $\varphi=\psi_1\land\psi_2$.

By definition:
$E_{\mathrm{can}}(\psi_1\land\psi_2)
\;=\;
E_{\mathrm{can}}(\psi_1)\;\bowtie\;
E_{\mathrm{can}}(\psi_2)$.

Choose $u\in E_{\mathrm{can}}(\psi_1)$ and $v\in E_{\mathrm{can}}(\psi_2)$.
By IH, $u\models\psi_1$ and $v\models\psi_2$.

Take any $e\in u\bowtie v$.

By the definition of $\bowtie$, $e$ is obtained from $u$ by:

- keeping $u$'s letters in their order,
- possibly inserting additional positions among them,
- and possibly \emph{merging} exactly one $v_j$ with the next unused $u_i$ by replacing
  $u_i$ with $u_i\cup v_j$,
- and never shrinking labels.

Thus every such $e$ is obtained from $u$ using only the operations:

\begin{itemize}
\item[(i)] insertion of new positions after the first,  
\item[(ii)] enlarging some letters.
\end{itemize}

These are exactly the conditions of Lemma~\ref{lem:mon-ext}.
Hence $u\leq_{hom} e$ and by \cref{lem:monotonicity} $e \models\psi_1$.

Symmetrically, since the definition of $\bowtie$ treats $u$ and $v$ symmetrically,
$e$ is also obtained from $v$ by insertions and enlargements only, and therefore
$e\models\psi_2$.

Thus $e\models\psi_1\land\psi_2$.
\end{proof}

Having established soundness, we now prove the converse direction:
every model of a formula homomorphically extends one of its canonical examples.
This shows that canonical examples are minimal with respect to
$\leq_{\mathrm{hom}}$.

\begin{lemma}[Canonical examples are $\leq_{\mathrm{hom}}$-minimal]
\label{lem:can-hom-min}
For every $\varphi \in \mathrm{LTL}_{\ltlFS,\land,\lor,\top,\bot }[\ap]$ and every ordinal word $w$,
if $w \models \varphi$ then there exists $e \in E_{\mathrm{can}}(\varphi)$ such that $e \leq_{\mathrm{hom}} w$.
\end{lemma}

\begin{proof}
By structural induction on $\varphi$.

\emph{Base cases.}
If $\varphi = \bot$, the premise $w\models\bot$ is impossible, so the statement is vacuous.

If $\varphi=\top$, let $e$ be the canonical word in $E_{\mathrm{can}}(\top) = \set{\emptyset}$  . Define $h(0)=0$.
Then $e[0]=\emptyset\subseteq w[0]$ and $h$ is strictly increasing , so
$e \leq_{\mathrm{hom}} w$.

If $\varphi=p\in\ap$, then $w\models p$ implies $p\in w[0]$.
Let $e$ be the canonical word in $E_{\mathrm{can}}(p)$, whose $0$-th letter is $\{p\}$ and all
other letters (if any) are empty. Define $h(0)=0$ and extend $h$ strictly along the domain of $e$
(e.g.\ $h(i)=i$ when the domain of $e$ is an initial segment).
Then $e[0]=\{p\}\subseteq w[0]$ and $h$ is strictly increasing, hence $e\leq_{\mathrm{hom}}w$.

\medskip
\emph{Step $\ltlFS$.}
Let $\varphi = \ltlFS\psi$ and assume $w\models\ltlFS\psi$.
By the semantics of strict $\ltlF$, there exists $i>0$ such that $w[i..] \models \psi$.

By the induction hypothesis applied to $\psi$ and $w[i..]$, there exist
$e' \in E_{\mathrm{can}}(\psi)$ and a strictly increasing homomorphism $h' : \mathrm{dom}(e') \to \mathrm{dom}(w[i..])$
such that $e' \leq_{\mathrm{hom}} w[i..]$ via $h'$.

By definition of canonical examples, $E_{\mathrm{can}}(\ltlFS\psi) = \set{\emptyset} \cdot E_{\mathrm{can}}(\psi)$.

Set $e := \emptyset \cdot e'$, i.e.\ $e[0]=\emptyset$ and for $\alpha>0$ we have
$e[\alpha] = e'[\alpha-1]$.

Define $h : \mathrm{dom}(e)\to\mathrm{dom}(w)$ by $h(0) := 0,
  \qquad
  h(1+\alpha) := i + h'(\alpha)$.

Then:
\begin{itemize}
  \item $h(0)=0$ and $h$ is strictly increasing since $i>0$ and $h'$ is strictly increasing;
  \item $e[0]=\emptyset\subseteq w[0]$;
  \item for $\alpha\ge 0$, $e[1+\alpha] = e'[\alpha]\subseteq w[i + h'(\alpha)] = w[h(1+\alpha)]$.
\end{itemize}

Thus $e \leq_{\mathrm{hom}} w$ and $e\in E_{\mathrm{can}}(\ltlFS\psi)$, completing this case.

\medskip
\emph{Step $\lor$.}
Let $\varphi = \psi_1 \lor \psi_2$ and suppose $w\models \psi_1 \lor \psi_2$.
Then $w\models\psi_i$ for some $i\in\{1,2\}$.
By IH (applied to $\psi_i$), there exists $e\in E_{\mathrm{can}}(\psi_i)$ with $e\leq_{\mathrm{hom}} w$.
Since $E_{\mathrm{can}}(\psi_1\lor\psi_2)
  = E_{\mathrm{can}}(\psi_1) \cup E_{\mathrm{can}}(\psi_2)$,
we have $e\in E_{\mathrm{can}}(\psi_1\lor\psi_2)$ and we are done.

\medskip
\emph{Step $\land$.}
Assume $\varphi = \psi_1 \land \psi_2$ and $w\models\psi_1\land\psi_2$.
By IH choose $e_1\in E_{\mathrm{can}}(\psi_1),
  e_2\in E_{\mathrm{can}}(\psi_2)$.
together with strictly increasing homomorphisms
\[
  h_1:dom(e_1) \to dom(w),
  \qquad
  h_2: dom(e_2)\to dom(w),
\]
witnessing $e_1\le_{\mathrm{hom}} w$ and $e_2\le_{\mathrm{hom}} w$.

Let $\Gamma := \mathrm{Im}(h_1)\cup \mathrm{Im}(h_2)$
and enumerate $\Gamma$ in strictly increasing order: $\gamma_0 < \gamma_1 < \cdots < \gamma_{L-1}$.

For each $k<L$ 
define: \\$i_k := \text{the unique } i \text{ with } h_1(i)=\gamma_k \text{ if such exists}$, $j_k := \text{the unique } j \text{ with } h_2(j)=\gamma_k \text{ if such exists}$,

and otherwise leave $i_k$ or $j_k$ undefined.
Since $h_1,h_2$ are strictly increasing, at most one $i_k$ and at most one $j_k$ exist.

Define a word $e$ of domain $L$ by
\[
  e[k] :=
  \begin{cases}
    e_1[i_k] & \text{if $i_k$ is defined and $j_k$ is not},\\[2pt]
    e_2[j_k] & \text{if $j_k$ is defined and $i_k$ is not},\\[2pt]
    e_1[i_k]\cup e_2[j_k] & \text{if both $i_k$ and $j_k$ are defined},\\[2pt]
    \emptyset & \text{if neither is defined (this never occurs for $k>0$)}.
  \end{cases}
\]

Because $h_1(0)=h_2(0)=0$, we have $i_0=j_0=0$, hence $e[0]=e_1[0]\cup e_2[0]$.

By construction, every tail index of $e_1$ appears exactly once among the $i_k$,
every tail index of $e_2$ appears exactly once among the $j_k$,
and the sequences $(i_k)$ and $(j_k)$ preserve order.
Thus $e \in e_1 \bowtie e_2
  \subseteq E_{\mathrm{can}}(\psi_1\land\psi_2)$.

Finally define $h(k) := \gamma_k$.

Then $h$ is strictly increasing, $h(0)=0$, and for every $k$:
- if $i_k$ is defined, then $e_1[i_k]\subseteq w[h_1(i_k)]=w[\gamma_k]=w[h(k)]$;
- if $j_k$ is defined, then $e_2[j_k]\subseteq w[h_2(j_k)]=w[\gamma_k]=w[h(k)]$.

Hence $e[k]\subseteq w[h(k)]$ for all $k$, so $e\le_{\mathrm{hom}} w$.

\medskip
This completes the conjunction case and the induction.
\end{proof}

Combining soundness, minimality, and monotonicity,
we obtain that the canonical set fully characterizes the formula.

\canonicalset*

\begin{proof}
Let $E^+_\varphi = \Ecan(\varphi)$.
By \cref{lem:can-sound,lem:can-hom-min,lem:monotonicity}, the claim follows.
\end{proof}

\subsection{Construction of maximal negative examples}

The maximal negative examples are obtained by dualising the canonical-set
construction.
Starting from the finite family of canonical positive examples of
$\varphi$, we build a finite set of words that systematically block every
possible homomorphic embedding of a positive witness.

Intuitively, the construction proceeds recursively.
At each stage, it selects letters that exclude the possible initial
matches of all positive witnesses, and then continues on the remaining
suffixes.
In this way, every word that does not satisfy $\varphi$ can be extended,
with respect to $\leq_{\hom}$, to one of the constructed examples.

A key difficulty is that embeddings may skip arbitrarily many positions.
For example, to block an embedding of
$\{p\}\{q\}\{p\}$, it is not enough to forbid a single occurrence of
$\{q\}$ after $\{p\}$: one must exclude all cases where $\{q\}$ appears
after an arbitrary finite prefix, and then ensure that no later position
matches the final $\{p\}$.
Since the skipped prefix can have any finite length, the construction must
range over all such lengths, which naturally leads to words of length up
to $\omega^2$.

By design, none of the constructed words admits an embedding of a
canonical positive example, and therefore none satisfies $\varphi$.
Hence they form a finite family of maximal counterexamples.

To formalize this blocking mechanism, we introduce a family of letters that simultaneously exclude prescribed matches, together with a canonical periodic word generated from them.

\begin{definition}[Exclusion-letter family and exclusion $\omega$-word]
    Let $\sigma_1,\ldots,\sigma_k \in 2^\ap$.
    Define
    \[
        \Excl(\sigma_1,\ldots,\sigma_k)
        =
        \Bigl\{
            \ap \setminus \{p_1,\ldots,p_k\}
            \;\Bigm|\;
            p_i \in \sigma_i\ \text{for all } i
        \Bigr\}.
    \]

    Enumerate $\Excl(\sigma_1,\ldots,\sigma_k)=\{B_1,\ldots,B_t\}$.
    The \emph{exclusion $\omega$-word} associated with $\sigma_1,\ldots,\sigma_k$ is $\ExclWord(\sigma_1,\ldots,\sigma_k)
        := (B_1 B_2 \cdots B_t)^\omega$.
    
\end{definition}

We now define the core recursive construction that generates
maximal negative examples for a given finite set of words.
The construction proceeds by gradually eliminating the possibility
of embedding each word in the set.

\begin{definition}[$\Dual^+$]
Let $A=\{w_1,\ldots,w_k\}$ be a finite set of nonempty finite words over $2^\ap$.
Define $\Dual^+(A)$ by induction on $N=\sum_{i=1}^k |w_i|$.

\smallskip\noindent\textbf{Base case.}
If some $w_i=\varepsilon$, set $\Dual^+(A)=\emptyset$.
If all $w_i$ have length~$1$, writing $w_i=\sigma_i$, set
\[
    \Dual^+(A)\ :=\ \ExclWord(\sigma_1,\ldots,\sigma_k).
\]

\smallskip\noindent\textbf{Inductive step.}
Assume some $w_i$ has length at least $2$, and write $w_i=\sigma_i\cdot w_i'$.
Let
\[
    \Excl(\sigma_1,\ldots,\sigma_k)=\{B_1,\ldots,B_t\},\qquad
    \ExclWord(\sigma_1,\ldots,\sigma_k)=(B_1\cdots B_t)^\omega.
\]
For every $\sigma\subseteq\ap$ such that
\[
   \forall j\in\{1,\ldots,t\}\;:\; \sigma\nsubseteq B_j
\]
define
\[
   \widehat{w_i}(\sigma)
   \;:=\;
   \begin{cases}
      w_i[1..] & \text{if }\sigma_i \subseteq \sigma,\\[3pt]
      w_i      & \text{otherwise}.
   \end{cases}
\]
Then
\[
    \Dual^+(A)
    \;=\;
    \bigcup_{\substack{\forall j:\sigma\nsubseteq B_j}}
        \ExclWord(\sigma_1,\ldots,\sigma_k)\cdot
        \sigma\cdot
        \Dual^+\!\bigl(\{\widehat{w_1}(\sigma),\ldots,\widehat{w_k}(\sigma)\}\bigr).
\]

\end{definition}

We first show that the construction is complete,
in the sense that every word that avoids all embeddings of $A$
can be extended to a word in $\Dual^+(A)$.

\begin{lemma}\label{lem:dual-plus-characterization}
Let $A=\{w_1,\ldots,w_k\}$ be a finite set of finite words over $2^\ap$
and let $u$ be an $\omega$-word.
If $w_i\nleq_{\mathrm{hom}^+}u$ for all $1\le i\le k$, 
then there exists $v\in\Dual^+(A)$ such that $u\leq_{\mathrm{hom}^+}v$.
\end{lemma}

\begin{proof}
We argue by induction on $N=\sum_{i=1}^k |w_i|$.

\medskip\noindent\textbf{Base case.}
If some $w_i=\varepsilon$, then $w_i\le_{\mathrm{hom}^+}u$ via the empty map,
contradicting the assumption that $w_i\nleq_{\mathrm{hom}^+}u$.
Hence this case cannot occur.

If all $w_i$ have length $1$, write $w_i=\sigma_i$.
Then $\Dual^+(A)=\ExclWord(\sigma_1,\ldots,\sigma_k)$.
Since $\sigma_i\nsubseteq u[j]$ for all $i,j$, every $u[j]$ belongs to
$\Excl(\sigma_1,\ldots,\sigma_k)$.
Let $v=(B_1\cdots B_t)^\omega$ be $\ExclWord(\sigma_1,\ldots,\sigma_k)$.

Define $h(-1)=-1$ and for $j\ge 0$, $h(j)=\min\{t>h(j-1)\mid u[j]\subseteq v[t]\}$. Because each $u[j]$ equals some $B_r$ and $v$ contains infinitely many $B_r$’s,
this defines a strictly increasing $h$ with $u[j]\subseteq v[h(j)]$.
Thus $u\le_{\mathrm{hom}^+}v$.

\medskip\noindent\textbf{Inductive step.}
Assume some $w_i$ has length at least $2$ and write $w_i=\sigma_i\cdot w_i'$.
Let $\Excl(\sigma_1,\ldots,\sigma_k)=\{B_1,\ldots,B_t\}$ and
$\ExclWord(\sigma_1,\ldots,\sigma_k)=(B_1\cdots B_t)^\omega$.

Define $J := \{ j\in\mathbb{N} \mid \exists i:\ \sigma_i \subseteq u[j] \}$.

\smallskip
\emph{Case 1: $J=\emptyset$.}
Write $\Excl(\sigma_1,\ldots,\sigma_k)=\{B_1,\ldots,B_t\}$. Then for every position $j$ there exists $\ell\in\{1,\ldots,t\}$ such that $u[j]\subseteq B_\ell$.
Hence, exactly as in the base case,
there exists a strictly increasing
$h$ such that $u \le_{\mathrm{hom}^+} \ExclWord(\sigma_1,\ldots,\sigma_k)$. We show that $\Dual^+(A)\neq\emptyset$.
Choose $\sigma\subseteq\ap$ such that $\forall \ell\in\{1,\ldots,t\}\quad \sigma\nsubseteq B_\ell$,
and such that every pruned word
\[
\widehat{w_i}(\sigma)=
\begin{cases}
w_i[1..], & \sigma_i\subseteq \sigma,\\
w_i, & \text{otherwise}
\end{cases}
\]
is nonempty.
This is possible: if $\sigma_i\subseteq\sigma$ would delete $w_i$ entirely,
then necessarily $|w_i|=1$, and we avoid this by choosing $\sigma$
so that $\sigma_i\nsubseteq\sigma$ for such indices $i$.
Since at least one word $w_i$ has length at least $2$, such a choice of
$\sigma$ exists.

Therefore, $\Dual^+\bigl(\{\widehat{w_1}(\sigma),\ldots,\widehat{w_k}(\sigma)\}\bigr)\neq\emptyset$. Pick any $v'\in \Dual^+\bigl(\{\widehat{w_1}(\sigma),\ldots,\widehat{w_k}(\sigma)\}\bigr)$, 
and define $v:=\ExclWord(\sigma_1,\ldots,\sigma_k)\cdot \sigma\cdot v'
\;\in\;\Dual^+(A)$.

Since the prefix of $v$ of length $\omega$ is
$\ExclWord(\sigma_1,\ldots,\sigma_k)$,
the same homomorphism $h$ also witnesses $u\le_{\mathrm{hom}^+} v$.

\smallskip\emph{Case 2: $J\neq\emptyset$.}
Let $j_0=\min J$ and $\sigma:=u[j_0]$.
Write $\Excl(\sigma_1,\ldots,\sigma_k)=\{B_1,\ldots,B_t\}.$

Then by the definition of $J$ we have $\forall \ell\in\{1,\ldots,t\}\quad \sigma \nsubseteq B_\ell$. Define $\widehat{w_i}(\sigma)$ as above.

\emph{Claim 1.} None of the $\widehat{w_i}(\sigma)$ is empty.  
Otherwise $w_i$ would be a one-letter word $\sigma_i$ with $\sigma_i\subseteq\sigma$,
hence $w_i\le_{\mathrm{hom}^+}u$ by mapping its unique position to $j_0$,
contrary to the hypothesis.

\emph{Claim 2.} For all $i$, $\widehat{w_i}(\sigma)\nleq_{\mathrm{hom}^+}u[(j_0 +1)..]$.
If $\widehat{w_i}(\sigma)\le_{\mathrm{hom}^+}u[(j_0 + 1)..]$ via $g$, then  
– if $\sigma_i\subseteq\sigma$, the first letter $\sigma_i$ embeds into $u[j_0]$
and composing this with $g$ yields $w_i\le_{\mathrm{hom}^+}u$;  
– if $\sigma_i\nsubseteq\sigma$, then $\widehat{w_i}(\sigma)=w_i$ and
$w_i\le_{\mathrm{hom}^+}u[(j_0 + 1)..] \implies w_i\le_{\mathrm{hom}^+}u$.
Both contradict $w_i\nleq_{\mathrm{hom}^+}u$.

\smallskip
Since at least one $\sigma_i\subseteq\sigma$, we have $\sum_i |\widehat{w_i}(\sigma)| < \sum_i |w_i| = N$. 
By Claim~2 the induction hypothesis applies to
$\widehat{A}(\sigma)=\{\widehat{w_1}(\sigma),\ldots,\widehat{w_k}(\sigma)\}
   \quad\text{and}\quad
   u[(j_0 + 1)..]$.
Hence there exists $v' \in \Dual^+(\widehat{A}(\sigma))
   \quad\text{with}\quad
   u[(j_0 + 1)..]\le_{\mathrm{hom}^+}v'$.

Let $v := \ExclWord(\sigma_1,\ldots,\sigma_k)\cdot\sigma\cdot v' \in \Dual^+(A)$.

Construct $h: dom(u) \to dom(v)$ by
\[
   h(j)=
   \begin{cases}
      h^{<}(j), & j<j_0,\\[2pt]
      \omega,   & j=j_0,\\[2pt]
      \omega+1+h^{>}(j-j_0-1), & j>j_0,
   \end{cases}
\]
where $h^{<}$ is the strictly increasing embedding of $u[0..j_0{-}1]$
into $\ExclWord(\sigma_1,\ldots,\sigma_k)$ constructed as in the base case,
and $h^{>}$ witnesses $u[(j_0 + 1)..]\le_{\mathrm{hom}^+}v'$.

The map $h$ is strictly increasing, and one checks directly that
$u[j]\subseteq v[h(j)]$ for all $j$.  
Thus $u\le_{\mathrm{hom}^+}v$.

\medskip
This completes the induction.
\end{proof}

The construction above handles strict embeddings.
To obtain the full dual characterization, we extend it
to arbitrary homomorphisms by exposing the first letter.

\begin{definition}[Dual]\label{def:dual}
Let $A=\{w_1,\ldots,w_k\}$ be a finite set of finite words.

For $\sigma\subseteq\ap$ and a word 
$w_i$,
define the pruned word
\[
   \widehat{w_i}(\sigma)
   :=
   \begin{cases}
      w_i[1..] & \text{if } w_i[0]\subseteq \sigma,\\[3pt]
      w_i      & \text{otherwise}.
   \end{cases}
\]
Thus $\widehat{w_i}(\sigma)$ is the suffix of $w_i$ obtained by removing 
its first letter when that letter embeds into $\sigma$; it may be empty.

The \emph{dual} of $A$ is the set
\[
   \Dual(A)
   :=
   \bigcup_{\sigma\subseteq\ap}
         \sigma \;\cdot\;
         \Dual^+(
            \{\widehat{w_1}(\sigma),\ldots,\widehat{w_k}(\sigma)\}
         \bigr),
\]
\end{definition}

We now show that the full dual construction captures all words
that avoid embeddings of $A$.

\begin{lemma}\label{lem:dual-characterization}
Let $A=\{w_1,\ldots,w_k\}$ be a finite set of finite words over $2^\ap$
and let $u$ be an $\omega$-word.
If $w_i\nleq_{\mathrm{hom}}u$ for all $1\le i\le k$, 
then there exists $v\in\Dual(A)$ such that $u\leq_{\mathrm{hom}}v$.
\end{lemma}

\begin{proof}
Let $\sigma := u[0]$ and let $u[1..]$ denote the suffix $u[1]u[2]\cdots$.

For each $i$, define
\[
   \widehat{w_i}(\sigma)
   :=
   \begin{cases}
      w_i[1..] & \text{if } w_i[0]\subseteq \sigma,\\[3pt]
      w_i      & \text{otherwise},
   \end{cases}
\]
and set $\widehat{A}(\sigma)
   :=\{\widehat{w_1}(\sigma),\ldots,\widehat{w_k}(\sigma)\}$.

\medskip\noindent
\emph{Claim 1: $\widehat{w_i}(\sigma)\neq\varepsilon$ for all $i$.}

If $\widehat{w_i}(\sigma)=\varepsilon$, then $w_i$ must be a one-letter word
with $w_i[0]\subseteq\sigma=u[0]$, and $w_i\leq_{\mathrm{hom}}u$ via the map
sending $0$ to $0$, contradicting the assumption $w_i\nleq_{\mathrm{hom}}u$.

\medskip\noindent
\emph{Claim 2: $\widehat{w_i}(\sigma)\nleq_{\mathrm{hom}}u[1..]$ for all $i$.}

Otherwise, suppose $\widehat{w_i}(\sigma)\leq_{\mathrm{hom}}u[1..]$ via some
strictly increasing $h$.
If $w_i[0]\subseteq\sigma$, then $w_i = w_i[0]\cdot w_i[1..]$ and we define
$\tilde{h}$ on $w_i$ by $\tilde{h}(0)=0,
   \qquad
   \tilde{h}(j+1)=1+h(j) \quad (j\ge 0)$.

Then $\tilde{h}$ is a homomorphism witnessing $w_i\leq_{\mathrm{hom}}u$,
again contradicting the hypothesis.
If $w_i[0]\not\subseteq\sigma$, then $\widehat{w_i}(\sigma)=w_i$ and
$w_i\leq_{\mathrm{hom}}u[1..]\leq_{\mathrm{hom}}u$, a contradiction as well.
Hence Claim~2 holds.

\medskip
By Claims~1–2, all words in $\widehat{A}(\sigma)$ are nonempty and satisfy
$\widehat{w_i}(\sigma)\nleq_{\mathrm{hom}}u[1..]$.
Thus we can apply Lemma~\ref{lem:dual-plus-characterization} to 
$\widehat{A}(\sigma)$ and $u[1..]$: there exists $v' \in \Dual^{+}(\widehat{A}(\sigma))
   \text{such that}
u[1..]\leq_{\mathrm{hom}}v'$.

By the definition of $\Dual(A)$ we then have $v := \sigma\cdot v' \in \Dual(A)$.

Let $h$ witness $u[1..]\leq_{\mathrm{hom}}v'$.
Define $\hat{h}$ by $\hat{h}(0) = 0,\quad
   \hat{h}(j) = 1 + h(j-1)\quad\text{for } j\ge 1$.
Clearly $\hat{h}$ is strictly increasing.
For $j\ge 1$ we have $u[j] = u[1..][j-1] \subseteq v'[h(j-1)] = v[\hat{h}(j)]$,
and for $j=0$ we have $u[0]=\sigma=v[0]=v[\hat{h}(0)]$.
Hence $\hat{h}$ witnesses $u\leq_{\mathrm{hom}}v$.

Since $v\in\Dual(A)$, the lemma follows.
\end{proof}

We next establish the soundness of the dual construction,
showing that no word in $\Dual^+(A)$ admits an embedding of any word in $A$.

\begin{lemma}\label{lem:dualplus-no-hom}
Let $A=\{w_1,\ldots,w_k\}$ be a finite set of nonempty finite words.
Then $\forall v\in \Dual^+(A),\ \forall i,\quad w_i \nleq_{\mathrm{hom}^+} v$.
\end{lemma}

\begin{proof}
We prove the statement by induction on 
\(
   N \ :=\ \sum_{i=1}^k |w_i|.
\)

\emph{Base 1:}  
If some $w_i=\epsilon$, then by definition $\Dual^+(A)=\emptyset$ and the claim is vacuous.

\emph{Base 2:}  
If all $w_i$ have length~$1$, say $w_i=\sigma_i$, then by the definition of 
\(\ExclWord(\sigma_1,\ldots,\sigma_k)\) every letter $B$ occurring in 
\(\ExclWord(\sigma_1,\ldots,\sigma_k)\) satisfies  
\(\sigma_i\not\subseteq B\).
Thus no homomorphism (even in the relaxed $\mathrm{hom}^+$ sense) can map 
$w_i$ into $\ExclWord(\sigma_1,\ldots,\sigma_k)$, proving the base.

\medskip
\emph{Inductive step.}
Assume some $w_i$ has length at least $2$.
Write each word as
\(
   w_i = \sigma_i\, w_i' ,
\)
where $\sigma_i$ is the first letter and $w_i'$ its suffix.

Let $v\in \Dual^+(A)$.
By the inductive definition of $\Dual^+(A)$, $v$ has the canonical form
\begin{equation}\label{eq:dualplus-decomp}
   v = \ExclWord(\sigma_1,\ldots,\sigma_k)\;\cdot\;\sigma\;\cdot\;v',
\end{equation}
where:
\begin{itemize}
    \item $\sigma \subseteq \ap$.
    \item $v' \in \Dual^+\!\left(\{\widehat{w_1}(\sigma),\ldots,\widehat{w_k}(\sigma)\}\right)$.
\end{itemize}
This decomposition follows directly from the inductive clause of $\Dual^+$:
the first block \(\ExclWord(\sigma_1,\ldots,\sigma_k)\) handles all first letters;
choosing a letter $\sigma$ determines which suffixes survive (via $\widehat{w_i}(\sigma)$),
and the remainder must lie in the recursive dual of these suffixes.

Now fix $i$.
Suppose towards contradiction that $w_i \le_{\mathrm{hom}^+} v$.
Let $h$ be such a homomorphism.

Since $h$ is strictly increasing and $w_i=\sigma_i w_i'$, we must have $\sigma_i \subseteq v[h(0)]$.
But by \eqref{eq:dualplus-decomp}, all letters in 
\(\ExclWord(\sigma_1,\ldots,\sigma_k)\) exclude every $\sigma_i$.
Hence $h(0)$ cannot lie in the $\ExclWord$-prefix, so necessarily $h(0) \geq |\ExclWord(\sigma_1,\ldots,\sigma_k)|$.
Thus the first letter of $w_i'$ must embed into $v[h(1)] \in v'$.
Consequently the suffix $w_i'$ embeds (via the shifted map $h'$)
into $v'$, i.e. $w_i' \le_{\mathrm{hom}^+} v'$.
But by construction, $v' \in \Dual^+\bigl(\{\widehat{w_1}(\sigma),\ldots,\widehat{w_k}(\sigma)\}\bigr)$ and $\widehat{w_i}(\sigma)=w_i'$.

The set of these suffixes has strictly smaller total length: $\sum_{j=1}^k |\widehat{w_j}(\sigma)| \ <\ N$.
Hence by the induction hypothesis, $w_i' \nleq_{\mathrm{hom}^+} v'$,
contradiction.

Thus $w_i \nleq_{\mathrm{hom}^+} v$ for every $i$, completing the induction.
\end{proof}

This property extends to the full dual construction,
yielding maximal negative examples.

\begin{lemma}\label{lem:dual-no-hom}
Let $A=\{w_1,\ldots,w_k\}$ be a finite set of finite words.
Then for every $v\in\Dual(A)$ and every $i$ we have $w_i \nleq_{\mathrm{hom}} v$.
\end{lemma}

\begin{proof}
Fix $v\in\Dual(A)$.

By \cref{def:dual}, there exists a letter $\sigma\subseteq\ap$ such that $v = \sigma \cdot v',
    \quad
    v' \in \Dual^{+}(A_\sigma)$,
where $A_\sigma := \{\widehat{w_1}(\sigma),\ldots,\widehat{w_k}(\sigma)\}$.

Assume for contradiction that $w_i \le_{\mathrm{hom}} v$ for some $i$.
Let $h$ witness this embedding.
Since $h$ is strictly increasing and $v[0]=\sigma$, we have $w_i[0]\subseteq\sigma$.
By the definition of pruning in \cref{def:dual}, this implies $\widehat{w_i}(\sigma) = w_i[1..]$.

The remainder of the embedding must map the suffix of $w_i$ into $v'$.

Hence $\widehat{w_i}(\sigma) = w_i[1..] \;\le_{\mathrm{hom}}\; v'$. But $\widehat{w_i}(\sigma)\in A_\sigma$ and $v'\in\Dual^{+}(A_\sigma)$.
By the lemma for $\Dual^{+}$ (see \cref{lem:dualplus-no-hom}), $\forall u\in A_\sigma,\ \forall z\in \Dual^{+}(A_\sigma),\quad u\nleq_{\mathrm{hom}} z$.
This contradicts $\widehat{w_i}(\sigma)\le_{\mathrm{hom}} v'$.

Therefore $w_i\nleq_{\mathrm{hom}} v$ for all $i$, completing the proof.
\end{proof}

We can now combine the completeness and soundness properties
of the dual construction to obtain a maximal set of negative examples.

\maxnegexamples*

\begin{proof}
Let \(E^-_\varphi = \Dual(\Ecan(\varphi))\).
By \cref{lem:dual-no-hom,lem:canonical-set},
every \(f \in E^-_\varphi\) satisfies \(f \not\models \varphi\).
Let \(w\) be any finite or $\omega$-length word with \(w \not\models \varphi\).
Since \(\Ecan(\varphi)\) is canonical, for all \(e \in \Ecan(\varphi)\) we have
\(e \not\leq_{\hom} w\).
By \cref{lem:dual-characterization}, there exists \(f \in E^-_\varphi\) such that
\(w \leq_{\hom} f\).
Thus \(F\) is maximal among negative example sets for \(\varphi\).
\end{proof}

Combining the canonical positive examples with the maximal negative
examples yields a finite labeled sample that completely determines the
semantics of the formula. We now use these ingredients to prove the following theorem.

\thmMain*

\begin{proof}
Let $\varphi \in \LTL_{\ltlFS,\land,\lor,\top,\bot}[\ap]$ and define the labeled sample
\[
S \;=\;
\{(e,+) \mid e \in E^+_\varphi
\;\cup\;
\{(f,-) \mid f \in E^-_\varphi \}.
\]
    
We show that $S$ uniquely characterizes $\varphi$.

Let $\psi \in \LTL_{\ltlFS,\land,\lor,\top,\bot}[\ap]$ be any formula that fits $S$.
We prove $\llbracket\varphi\rrbracket = \llbracket\psi\rrbracket$.

\smallskip
\noindent
($\subseteq$)
Let $w \models \varphi$.  
By \cref{lem:canonical-set}, $E^+_\varphi$ is a canonical set, hence there exists
$e \in E^+_\varphi$ such that $e \leq_{\mathrm{hom}} w$.
Since $\psi$ fits $S$, we have $e \models \psi$.
By monotonicity it follows that
$w \models \psi$.

\smallskip
\noindent
($\supseteq$)
Assume towards a contradiction that there exists a word
$w \in \llbracket\psi\rrbracket \setminus \llbracket\varphi\rrbracket$.
By \cref{thm:finite-infinite-equivalence}, we may assume that $w$ has finite length.
Since $w \not\models \varphi$, 
by \cref{lem:max-negative-examples}, there exists
$f \in E^-_\varphi$ such that $w \leq_{\mathrm{hom}} f$.
As $w \models \psi$, monotonicity implies $f \models \psi$.
However, $(f,-) \in S$ and $\psi$ fits $S$, a contradiction.

\smallskip
\noindent
Therefore, $\llbracket\varphi\rrbracket = \llbracket\psi\rrbracket$, and $\varphi$ is uniquely characterized by $S$.
\end{proof}

Having established finite characterizations for the fragment
$\LTL_{\ltlFS,\land,\lor,\top,\bot}$,
we conclude by showing that such characterizations do not extend
to richer fragments.

\thmXFFSneg*
\begin{proof}
$\LTL   _{\ltlX,\land,\lor,\neg}$:
Consider the formula $p$. Suppose for contradiction that $p$ has a finite characterization $(E^+,E^-)$ with $E^-=\{w_1,\ldots,w_k\}$. Let
\[\varphi:=\underset{p\not\in w_i(i+1)}{\underset{1\leq i\leq k}{\bigwedge}}\ltlX^{i+1}p\wedge\underset{p\in w_i(i+1)}{\underset{1\leq i\leq k}{\bigwedge}}\ltlX^{i+1}\neg p\]
and take $\psi:=p\vee \varphi$. We claim that $\psi$ fits $(E^+,E^-)$. First of all, clearly $\psi$ fits $E^+$, because $E^+$ consists of examples that are positive for $p$ and $p$ entails $\psi$. Next, we show that for all $1\leq i\leq k$, $w_i\not\models\psi$. As $w_i\in E^-$ is a negative example of $p$, 
clearly $w_i\not\models p$. 
Furthermore, by construction, $\varphi$ is false in every negative example. Therefore, 
$\psi$ is also false in every negative example.
Finally, we must show that $\psi$ is not equivalent to $p$. Indeed, let $w$ be
the finite word of length  $k+1$  with $w[0]=\emptyset$ and 
$w[i+1]=\ap\setminus w_i[i+1]$.
Then $w$ makes $\psi$ true but not $p$.

$\LTL_{\ltlF,\land,\lor,\neg}$:
Consider the formula $p$. Suppose for contradiction that $p$ has a finite characterization $(E^+,E^-)$ with $E^-=\{w_1,\ldots,w_n\}$. For each
$i\leq n$, let $k_i\in\mathbb{N}\cup\{\infty\}$ be the number of 
times the truth value of $p$ changes in $w_i$. Now, 
let $k$ be any natural number different from $k_1, \ldots, k_n$. For $m\in\mathbb{N}$, let 
\[\begin{array}{lll}
\phi_{\geq 0} &=& \top \\
\phi_{\geq m+1} &=& (p \land \ltlF(\neg p\land \phi_{\geq m}))\lor (\neg p\land \ltlF(p\land \phi_{\geq m})) \\
\phi_{=m} &=& \phi_{\geq m}\land\neg\phi_{\geq m+1}
\end{array}\]
Observe that $\phi_{=m}$ is true in a word 
precisely if the truth value of $p$ changes exactly $m$ times. Now, 
take $\psi:=p\vee \phi_{=k}$. We claim that $\psi$ fits $(E^+,E^-)$. First of all, clearly $\psi$ fits $E^+$, because $E^+$ consists of examples that are positive for $p$ and $p$ entails $\psi$. Furthermore, for all $1\leq i\leq n$, $w_i\not\models\psi$: as $w_i\in E^-$ is a negative example of $p$, clearly $w\not\models p$. And, by construction, $\phi$ is false in every negative example. Therefore, 
$\psi$ is also false in every negative example.
Finally, we must show that $\psi$ is not equivalent to $p$. Indeed, let $w$ be
any word in which $p$ is initially false and changes truth value exactly $m$ times. Then $w$ makes $\psi$ true but not $p$.

$\LTL_{\ltlFS,\land,\lor,\neg}$:
follows immediately from the above since
$\ltlF\phi$ is definable as 
$\ltlFS\phi\lor\phi$.
\end{proof}

\thmFXandOr*
\begin{proof}
Let $\phi := \ltlF(p \land q)$ and, for each $n \in \mathbb{N}$, define
\[
\psi_n := \phi \;\lor\; \bigl(p \land \ltlX p \land \ltlX^2 p \land \cdots \land \ltlX^n p \land \ltlX^{n+1} q\bigr).
\]
Observe that $\sema{\phi} \subseteq \sema{\psi_n}$ for all $n$.

Assume, towards a contradiction, that there exists a finite set $S$ of (transfinite) examples that distinguishes $\phi$ from every $\psi_n$. Since $S$ is finite, by the pigeonhole principle there exists an example $(w,b) \in S$ that distinguishes $\phi$ from infinitely many $\psi_n$. Hence, for infinitely many $n$, we have $w \models \psi_n
\quad\text{and}\quad
w \not\models \phi$.

For each such $n$, since $w \not\models \phi$, it must satisfy the second disjunct of $\psi_n$, i.e.,
\[
w \models p \land \ltlX p \land \cdots \land \ltlX^n p \land \ltlX^{n+1} q.
\]
Thus, for infinitely many $n$, the word $w$ contains a position $n+1$ where $q$ holds, preceded by a prefix of $p$'s. It follows that there exist at least two distinct indices $m \neq n$ such that $w \models \psi_m$ and $w \models \psi_n$. This implies that $w$ contains two positions where $q$ holds, each preceded by a (possibly different) finite prefix of $p$'s, and hence $w \models \ltlF(p \land q) = \phi$, a contradiction.

Therefore, no finite set of examples can distinguish $\phi$ from all $\psi_n$, and thus $\LTL_{\ltlF,\ltlX,\wedge,\vee}[\ap]$ does not admit finite characterizations over $\allwords$.
\end{proof}

\corDual*
\begin{proof}
    Let $\phi\in\LTL_{\ltlGS,\land,\lor,\top,\bot}[\ap]$ be given. Let $\phi^\neg$ be obtained from 
    $\phi$ by replacing every atomic proposition $p$ by $\neg p$.    
    Observe that $\neg\phi^\neg$ can be equivalently written as an $\LTL_{\ltlFS,\land,\lor,\top,\bot}[\ap]$-formula $\chi$. Then $\chi$ is uniquely characterized, relative to $\LTL_{\ltlFS,\land,\lor,\top,\bot}[\ap]$, by a finite set of labeled examples $E_\chi^-\cup E_\chi^+$.
    For each $\ap$-word $w$, let $\widehat{w}$ be the 
    corresponding $\ap$-word, where 
    each letter is replaced by its complement (relative to $\ap$).
    Let $E^+_\phi=\{\widehat{w}\mid w\in E^-_\chi\}$ and $E^-_\phi = \{\widehat{w}\mid w\in E^+_\chi\}$.

    Claim 1: $\phi$ fits $(E^+_\phi,E^-_\phi)$.  Indeed,
    suppose $\widehat{w}\not\models\phi$ for some $\widehat{w}\in E^+_\phi$.  Then 
    $\widehat{w}\models\neg\phi$ and hence $w\models\chi$ but $w\in E^-_\chi$, a contradiction. Similarly for the negative examples. 

    Claim 2: $(E^+_\phi,E^-_\phi)$ uniquely characterizes $\phi$.  
    Indeed, suppose $\psi$ fits. 
    Again $\neg\psi^\neg$ can be written equivalently as an $\LTL_{\ltlFS,\land,\lor,\top,\bot}[\ap]$-formula $\theta$. By the same reasoning as before, we can show that
    $\theta$ fits $(E^+_\chi,E^-_\chi)$. Hence, $\theta$ is 
    equivalent to $\chi$, from which it follows that $\psi$ is equivalent to $\phi$.
\end{proof}

\subsection{Proofs of the complexity bounds}

This section provides the proofs of the bounds stated in Remark~\ref{rem:bound}.
The analysis proceeds in two stages: first, by controlling the length
and number of canonical examples, and second, by bounding the growth
of the dual construction derived from them.

\paragraph*{Structural Measures}

We introduce two structural parameters of formulas.

\paragraph*{$\ltlFS$-count.}
Let $\numOfFS(\varphi)$ denote the total number of $\ltlFS$ operators
in the syntax tree of $\varphi$:
\[
\begin{aligned}
\numOfFS(\bot) &= 0, \\
\numOfFS(\top) &= 0, \\
\numOfFS(p) &= 0, \\
\numOfFS(\ltlFS\varphi) &= \numOfFS(\varphi) + 1, \\
\numOfFS(\varphi \circ \psi)
&=
\numOfFS(\varphi)+\numOfFS(\psi)
\quad (\circ \in \{\land,\lor\}).
\end{aligned}
\]

\paragraph*{Leaf count.}
Let $\numOfleaf(\varphi)$ denote the total number of leaves in the syntax tree of $\varphi$: 
\[
\begin{aligned}
\numOfleaf(\bot) &= 1, \\
\numOfleaf(\top) &= 1, \\
\numOfleaf(p) &= 1, \\
\numOfleaf(\ltlFS \varphi) &= \numOfleaf(\varphi), \\
\numOfleaf(\varphi \circ \psi)
&=
\numOfleaf(\varphi)+\numOfleaf(\psi)
\quad (\circ \in \{\land,\lor\}).
\end{aligned}
\]

For a finite set of words $E$, define $\ell_{\max}(E)
:=
\max\{\, |u| \mid u \in E \,\},
\quad
n(E)
:=
|E|$.

Let $m := |\varphi|$. These structural measures provide a convenient way to relate the syntax
of a formula to the size of its associated example sets.
In particular, they allow the length and number of canonical examples
to be bounded in terms of the size of the formula.

\begin{proposition}[Length bound]\label{prop:length-bound}
For every formula $\varphi$,
\[
\ell_{\max}(\Ecan(\varphi))
\le
\numOfFS(\varphi)+\numOfleaf(\varphi).
\]
In particular, if $m:=|\varphi|$, then $\ell_{\max}(\Ecan(\varphi))\le m$.
\end{proposition}

\begin{proof}
We prove the bound by structural induction on $\varphi$.

\textbf{Base cases.}
If $\varphi\in\{\top,p\}$ then $\Ecan(\varphi)$ contains a single word of
length $1$, hence $\ell_{\max}(\Ecan(\varphi))=1
=
\numOfleaf(\varphi)+\numOfFS(\varphi)$.

If $\varphi=\bot$, then $\Ecan(\varphi)=\emptyset$ and
$\ell_{\max}(\Ecan(\varphi))=0$, so the bound holds.

\textbf{$\ltlFS$ case.}
Suppose $\varphi=\ltlFS\psi$. Since $\Ecan(\varphi)=\set{\emptyset}\cdot\Ecan(\psi)$, each canonical word gains one prefix letter and therefore $\ell_{\max}(\Ecan(\varphi)) = \ell_{\max}(\Ecan(\psi))+1$.
By the induction hypothesis, $\ell_{\max}(\Ecan(\psi))
\le
\numOfleaf(\psi)+\numOfFS(\psi)$, hence $\ell_{\max}(\Ecan(\varphi))
\le
\numOfleaf(\varphi)+\numOfFS(\varphi)$.

\textbf{$\lor$ case.}
Suppose $\varphi$ is a disjunction with subformulas $\psi_1,\psi_2$.

Then, $\Ecan(\varphi)=\Ecan(\psi_1)\cup\Ecan(\psi_2)$, so $\ell_{\max}(\Ecan(\varphi)) = \max\{\ell_{\max}(\Ecan(\psi_1)),\ell_{\max}(\Ecan(\psi_2))\}$. Applying the induction hypothesis gives $\ell_{\max}(\Ecan(\varphi))
\le
\numOfleaf(\varphi)+\numOfFS(\varphi)$.

\textbf{$\land$ case.}
Suppose $\varphi$ is a conjunction with subformulas $\psi_1,\psi_2$.
For $u\in\Ecan(\psi_1)$ and $v\in\Ecan(\psi_2)$, every
$e\in u\mergeop v$ satisfies $|e|\le |u|+|v|-1$.

Taking maxima yields $\ell_{\max}(\Ecan(\varphi))
\le
\ell_{\max}(\Ecan(\psi_1))
+
\ell_{\max}(\Ecan(\psi_2))$. Using the induction hypothesis, $\ell_{\max}(\Ecan(\varphi))
\le
\numOfleaf(\varphi)+\numOfFS(\varphi)$.

Thus the bound holds for all formulas. Since the size
$m=|\varphi|$ counts both leaves and $\ltlFS$ occurrences, $\numOfleaf(\varphi)+\numOfFS(\varphi)\le m$, and therefore $\ell_{\max}(\Ecan(\varphi))\le m$.
\end{proof}

\begin{proposition}\label{prop:merge-bound}
Let $\varphi_1,\varphi_2 \in \LTL_{\ltlFS,\land,\lor,\top,\bot}[\ap]$ and define 
\[
n_i := |\Ecan(\varphi_i)|, 
\quad
\ell_i := \ell_{\max}(\Ecan(\varphi_i))
\quad (i=1,2)
\].

Then,
$|\Ecan(\varphi_1 \land \varphi_2)|
\le
n_1 n_2 \, D(\ell_1-1,\ell_2-1)
\le
n_1 n_2 \, 3^{\ell_1+\ell_2}$.
\end{proposition}

\begin{proof}
By definition, $\Ecan(\varphi_1 \land \varphi_2)
=
\Ecan(\varphi_1) \mergeop \Ecan(\varphi_2)
=
\bigcup_{u\in\Ecan(\varphi_1)}
\bigcup_{v\in\Ecan(\varphi_2)}
(u \mergeop v)$.

Hence, $|\Ecan(\varphi_1 \land \varphi_2)|
\le
\sum_{u\in\Ecan(\varphi_1)}
\sum_{v\in\Ecan(\varphi_2)}
|u \mergeop v|$.

For fixed words $u,v$, merged interleavings correspond to lattice paths from
$(0,0)$ to $(|u|-1,|v|-1)$ with steps $(1,0)$, $(0,1)$, and $(1,1)$.
Thus, $|u \mergeop v| = D(|u|-1,|v|-1)$.
Since $|u|\le \ell_1$ and $|v|\le \ell_2$, monotonicity of Delannoy numbers gives $|u \mergeop v|
\le
D(\ell_1-1,\ell_2-1)$. Therefore $|\Ecan(\varphi_1 \land \varphi_2)|
\le
n_1 n_2\, D(\ell_1-1,\ell_2-1)$. Using the standard bound $D(a,b)\le 3^{a+b}$ yields $D(\ell_1-1,\ell_2-1)
\le
3^{(\ell_1-1)+(\ell_2-1)}
\le
3^{\ell_1+\ell_2}$, which proves the claim.
\end{proof}

\begin{lemma}\label{lem:ecan-size-exp}
Let $\varphi \in \LTL_{\ltlFS,\land,\lor,\top,\bot}[\ap]$ and let $m := |\varphi|$ be the size of its syntax tree. Then there exists a constant $c>0$ such that $|\Ecan(\varphi)| \le 3^{cm}$.
In particular, $|\Ecan(\varphi)| = O(2^m)$.
\end{lemma}

\begin{proof}
We prove by structural induction on $\varphi$ that $|\Ecan(\varphi)| \le 3^{cm}$
for some constant $c>0$.

\textbf{Base cases.}
If $\varphi \in \{p,\top,\bot\}$, then $m=1$ and $|\Ecan(\varphi)| \le 1 \le 3^c$.

\textbf{$\ltlFS$ case.}
Let $\varphi=\ltlFS\psi$ and let $m_1:=|\psi|$. Then $m=m_1+1$. Since $\Ecan(\varphi)=\set{\emptyset}\cdot\Ecan(\psi)$, we have $|\Ecan(\varphi)|=|\Ecan(\psi)|$. By the induction hypothesis, $|\Ecan(\varphi)|
=
|\Ecan(\psi)|
\le
3^{cm_1}
\le
3^{cm}$.

\textbf{$\lor$ case.}
Let $\varphi=\psi_1\lor\psi_2$ and let $m_i:=|\psi_i|$ for $i\in\{1,2\}$. Then
$m=m_1+m_2+1$. Since $\Ecan(\varphi)=\Ecan(\psi_1)\cup\Ecan(\psi_2)$,
we obtain $|\Ecan(\varphi)|
\le
|\Ecan(\psi_1)|+|\Ecan(\psi_2)|$.
By the induction hypothesis, $|\Ecan(\varphi)|
\le
3^{cm_1}+3^{cm_2}
\le
2\cdot 3^{c(m_1+m_2)}
\le
3^{cm}$, for sufficiently large $c$.

\textbf{$\land$ case.}
Let $\varphi=\psi_1\land\psi_2$ and let $m_i:=|\psi_i|$ for $i\in\{1,2\}$. Then
$m=m_1+m_2+1$. By Proposition~\ref{prop:merge-bound} and Proposition~\ref{prop:length-bound}, $|\Ecan(\varphi)|
\le
|\Ecan(\psi_1)|\,|\Ecan(\psi_2)|\,3^{m_1+m_2}$.
Using the induction hypothesis, $|\Ecan(\varphi)|
\le
3^{cm_1}\cdot 3^{cm_2}\cdot 3^{m_1+m_2}
=
3^{(c+1)(m_1+m_2)}
\le
3^{(c+1)m}$. Increasing $c$ if necessary completes the induction.
\end{proof}

\medskip

Having bounded the size of the canonical set, the focus shifts to the
dual construction.
Since the dual is defined recursively over the canonical examples,
its size depends both on the number and the length of these words,
as well as on the size of the alphabet.

\begin{proposition}
Let \(A\) be a finite set of nonempty finite words over \(2^{\ap}\). Write
$h:=|\ap|$,
$ k:=\sum_{w\in A}|w|$.
Then
\(
|\Dual^+(A)| \le 2^{hk}
\)
and
\(
|\Dual(A)| \le 2^{h(k+1)}.
\)
\end{proposition}

\begin{proof}
We first prove by induction on
\(
k=\sum_{w\in A}|w|
\)
that
$
|\Dual^+(A)| \le 2^{hk}.
$

\textbf{Base case.}
Assume \(k=0\). Then some word in \(A\) is empty, and by definition
$
\Dual^+(A)=\emptyset.
$
Hence
$
|\Dual^+(A)|=0\le 1=2^{h\cdot 0}.
$

\textbf{Inductive step.}
Assume the claim holds for all finite sets \(A'\) of finite words such that
$
\sum_{w\in A'}|w|<k,
$
and let \(A=\{w_1,\dots,w_n\}\) satisfy
$
\sum_{i=1}^n |w_i|=k.
$
If some \(w_i=\varepsilon\), then again \(\Dual^+(A)=\emptyset\), so the claim is immediate.

Otherwise, if all words have length \(1\), then by definition \(\Dual^+(A)\) is a single \(\omega\)-word, and therefore
$
|\Dual^+(A)|=1\le 2^{hk}.$

It remains to consider the case where no word is empty and at least one word
has length at least \(2\).
Write
$w_i=\sigma_i\cdot w_i'$ for 
$i=1,\dots,k$,
where \(\sigma_i\in 2^{\ap}\).
Let
$\Excl(\sigma_1,\ldots,\sigma_k)=\{B_1,\ldots,B_t\}.$
Then
\[
\Dual^+(A)
=
\bigcup_{\forall j\in\{1,\ldots,t\}:\ \sigma\nsubseteq B_j}
\ExclWord(\sigma_1,\ldots,\sigma_k)\cdot
\sigma\cdot
\Dual^+\!\bigl(\{\widehat{w_1}(\sigma),\ldots,\widehat{w_k}(\sigma)\}\bigr).
\]
Hence, $|\Dual^+(A)|
\le
\sum_{\sigma\subseteq\ap}
\left|
\Dual^+\!\bigl(\{\widehat{w_1}(\sigma),\dots,\widehat{w_n}(\sigma)\}\bigr)
\right|$.

There are at most \(2^h\) choices of \(\sigma\subseteq\ap\). Now for each such \(\sigma\), since \(\sigma\notin\Excl(\sigma_1,\dots,\sigma_n)\), there exists some \(i\) such that \(\sigma_i\subseteq \sigma\). For this \(i\), the word \(\widehat{w_i}(\sigma)=w_i[1..]\) is shorter than \(w_i\). Thus, $\sum_{j=1}^n |\widehat{w_j}(\sigma)| \le k-1$. 

Therefore, by the induction hypothesis, $\left|
\Dual^+\!\bigl(\{\widehat{w_1}(\sigma),\dots,\widehat{w_n}(\sigma)\}\bigr)
\right|
\le
2^{h(k-1)}$. It follows that $|\Dual^+(A)|
\le
2^h\cdot 2^{h(k-1)}
=
2^{hk}$. This completes the induction for \(\Dual^+(A)\).

\medskip

We now bound \(\Dual(A)\). By definition,
\[
\Dual(A)
=
\bigcup_{\sigma\subseteq\ap}
\sigma\cdot
\Dual^+\!\bigl(\{\widehat{w_1}(\sigma),\dots,\widehat{w_n}(\sigma)\}\bigr).
\]
Thus, $|\Dual(A)|
\le
\sum_{\sigma\subseteq\ap}
\left|
\Dual^+\!\bigl(\{\widehat{w_1}(\sigma),\dots,\widehat{w_n}(\sigma)\}\bigr)
\right|$. 

For each \(\sigma\subseteq\ap\), $\sum_{i=1}^n |\widehat{w_i}(\sigma)| \le k$, so by the bound already proved for \(\Dual^+\), $\left|
\Dual^+\!\bigl(\{\widehat{w_1}(\sigma),\dots,\widehat{w_n}(\sigma)\}\bigr)
\right|
\le
2^{hk}$.

Since there are \(2^h\) choices of \(\sigma\subseteq\ap\), we obtain $|\Dual(A)|
\le
2^h\cdot 2^{hk}
=
2^{h(k+1)}$.
\end{proof}

The bounds established so far can now be combined to obtain an overall
estimate on the number of the negative examples derived from a formula.
This yields the desired complexity bound for the full characterization.

\begin{lemma}\label{lem:dual-ecan-bound}
Fix a finite set $\ap$ and
let \(\varphi \in \LTL_{\ltlFS,\land,\lor,\top,\bot}[\ap]\).
Let
$m:=|\varphi|$,
$h:=|\ap|$.
Then
\(
\bigl|\Dual(\Ecan(\varphi))\bigr|
\le
2^{\,O(m2^m)}.
\)
More precisely, there exists a constant \(c>0\) such that
\(
\bigl|\Dual(\Ecan(\varphi))\bigr|
\le
2^{\,h(cm2^m+1)}.
\)
Moreover, each word in $\Dual(\Ecan(\varphi))$ is given by a word expression of size at most $O(mh2^{m+h})$. 
\end{lemma}

\begin{proof}
Set $E:=\Ecan(\varphi), \quad k:=\sum_{w\in \Ecan(\varphi)}|w|$

By the bound for \(\Dual\), $|\Dual(E)| \le 2^{h(k+1)}$. It therefore suffices to bound \(k\). 

Since
$k \le n(E)\cdot \ell_{\max}(E)$, by Proposition~\ref{prop:length-bound} we have $\ell_{\max}(E)=\ell_{\max}(\Ecan(\varphi)) \le m$, and by Lemma~\ref{lem:ecan-size-exp} there exists \(c>0\) such that $n(E)=|\Ecan(\varphi)| \le c\,2^m$. Hence $k \le c\,m\,2^m$. Substituting into the bound for \(\Dual(E)\), $\bigl|\Dual(\Ecan(\varphi))\bigr|
\le
2^{\,h(cm2^m+1)}$.

To see why each word in $\Dual(E)$ is given by a word expression of size at most $O(mh2^{m+h})$, note that $\ExclWord(B_1,\ldots,B_l)$ is of size at most $O(h2^h)$ and also at each step of the construction of $\Dual^+(A)$ at least one of the input word's length decreases by 1 and therefore the number of steps is bounded by $k$, and because $k = O(m2^m)$ the bound holds. 
\end{proof}

This completes the analysis of the cardinality of the characterization
and establishes the claimed complexity bounds.

We are now ready to establish the bounds stated in Remark~\ref{rem:bound}. The claimed bounds follow directly from the estimates established in \cref{lem:dual-ecan-bound,lem:ecan-size-exp}.

\section{Proofs for Section~\ref{sec:schematic-examples}}

The remainder of this section lifts the finite-characterization framework to a more compact
representation based on regular expressions.
The goal is to replace infinite $\omega$-behavior by finite schematic
descriptions that preserve satisfaction for the fragment under consideration.

To achieve this, every regular transfinite word is associated with
a finite regular expression that abstracts away infinite repetition
while retaining its essential combinatorial structure.

\begin{definition}[Translation of a finitely presented transfinite word into a schematic example]\label{def:schematic-example}
Let $\ap$ be finite, and let $w$ be a transfinite word over
$\Sigma=2^{\ap}$. Assume that $w$ is represented by an expression $e$ built from letters in $\Sigma$ using concatenation and
$\omega$-power.

The \emph{schematic example} associated with $w$, denoted $r(w)$, is the
union-free regular expression over the alphabet $\mathbb{B}(\ap)$ obtained
from $e$ by recursively replacing each letter $\sigma \in 2^{\ap}$ by a Boolean
expression describing exactly $\sigma$, and replacing every occurrence of
$\omega$ by Kleene star:
\[
\begin{aligned}
r(\sigma) &:= \Bigl(\bigwedge_{p\in \sigma} p\Bigr)\wedge
        \Bigl(\bigwedge_{q\notin \sigma}\neg q\Bigr),\\
r(e_1e_2) &:= r(e_1)\,r(e_2),\\
r(e^\omega) &:= (r(e))^* .
\end{aligned}
\]
Thus each letter $\sigma \subseteq\ap$ is viewed as the unique valuation satisfying
$r(\sigma)$.
\end{definition}
A key property of the fragment
$\LTL_{\ltlFS,\land,\lor,\top,\bot}[\ap]$
is that replacing $\omega$-powers by finite repetition
does not affect satisfiability.
This is formalised by the following equivalence.

\begin{lemma}\label{lem:omega-to-star}
Let $w$ be a transfinite word over $2^{\ap}$, and let
$\varphi \in\LTL_{\ltlFS,\land,\lor,\top,\bot}[\ap]$.
Then
\[
w \models \varphi
\quad\iff\quad
\exists u \in L(r(w)) \text{ such that } u \models \varphi .
\]
\end{lemma}

\begin{proof}
We prove the two directions separately.

For the implication from right to left, let $u\in L(r(w))$ and assume
$u\models \varphi$.
By construction, $u$ is obtained from a regular presentation of $w$ by
replacing each occurrence of $\omega$-power by a finite number of repetitions.
Hence $u$ homomorphically embeds into $w$.
Since formulas in
$\LTL_{\ltlFS,\land,\lor,\top,\bot}[\ap]$
are monotone with respect to homomorphic embedding, it follows that
$w\models \varphi$.

For the converse implication, we prove by induction on the structure of
$\varphi$.

\medskip
The cases $\top$ and $\bot$ are immediate.
For an atomic proposition $p$, if $w\models p$, then $p$ holds at the first
position of $w$. Choosing any finite unfolding of $w$ that preserves this first
position yields some $u\in L(r(w))$ with $u\models p$.

If $\varphi=\psi\vee\chi$ and $w\models\psi\vee\chi$, then either
$w\models\psi$ or $w\models\chi$.
The induction hypothesis applied to the satisfied disjunct yields
$u\in L(r(v))$ with $u\models\psi$ or $u\models\chi$, and therefore
$u\models\psi\vee\chi$.

If $\varphi=\psi\wedge\chi$ and $w\models\psi\wedge\chi$, then
$w\models\psi$ and $w\models\chi$.
By induction hypothesis there exist
$u_\psi,u_\chi\in L(r(v))$
such that
$u_\psi\models\psi$ and $u_\chi\models\chi$.
Both $u_\psi$ and $u_\chi$ are obtained by choosing finite numbers of
repetitions for the $\omega$-powers occurring in the presentation of $v$.
Choose a single finite unfolding $u\in L(r(v))$ in which each such
$\omega$-power is repeated at least as many times as in both
$u_\psi$ and $u_\chi$.
Then
$u_\psi \le_{\hom} u$
and
$u_\chi \le_{\hom} u$.
By monotonicity,
$u\models\psi$ and $u\models\chi$,
hence
$u\models\psi\wedge\chi$.

Finally, let $\varphi=\ltlFS\psi$ and assume $v\models\ltlFS\psi$.
Then there exists a position $\beta>0$ such that the suffix $w[\beta..]$
satisfies $\psi$.
Because $w$ is a  transfinite word, the suffix
$w[\beta..]$ is again regular and finitely presented.
By induction hypothesis, there exists
$u' \in L(r(w[\beta..]))$
such that
$u'\models\psi$.
Now $\beta$ is reached after finitely many traversals of the $\omega$-powers
appearing in the presentation of $w$, so there exists a finite prefix $x$
such that $x\cdot u' \in L(r(w))$.
Since $u'$ begins strictly after the first position, we obtain $x\cdot u' \models \ltlFS\psi$.
Thus there exists $u\in L(r(w))$ such that $u\models\ltlFS\psi$.

This completes the proof.
\end{proof}
The lemma establishes that schematic examples faithfully capture
the behavior of transfinite words with respect to the fragment.

The next result shows that replacing words in a sample by their
schematic representations preserves the ability of the sample
to characterize a formula.

\begin{lemma}\label{lem:schematic-correctness}
Let $\varphi \in \LTL_{\ltlFS,\land,\lor,\top,\bot}[\ap]$.
Let $S \subseteq (2^\ap)^{\mathrm{all}} \times \{-,+\}$ be a sample such that
every positively labeled word is \emph{finite}.
Define
\(
Schematic(S)
\;:=\;
\{(r(w),b) \mid (w,b) \in S\}.
\)
Then
\[
S \text{ fits } \varphi
\quad\iff\quad
Schematic(S) \text{ fits } \varphi.
\]
\end{lemma}

\begin{proof}
($\Rightarrow$)
Assume that \(S\) fits \(\varphi\), and let \((w,b)\in S\).

If \(b=+\), then \(w \models \varphi\). Since \(w\) is finite, we have
\(L(r(w))=\{w\}\). Hence there exists a word in \(L(r(w))\) satisfying
\(\varphi\), namely \(w\) itself, so \((r(w),+)\) fits \(\varphi\).

If \(b=-\), then \(w \not\models \varphi\). Suppose toward a contradiction that
there exists \(u\in L(r(w))\) such that \(u \models \varphi\). By
Lemma~\ref{lem:omega-to-star}, this would imply \(w \models \varphi\), a contradiction.
Hence no word in \(L(r(w))\) satisfies \(\varphi\), and therefore \((r(w),-)\)
fits \(\varphi\).

\smallskip
($\Leftarrow$)
Assume that \(Schematic(S)\) fits \(\varphi\), and let \((w,b)\in S\).

If \(b=+\), then \((r(w),+)\) fits \(\varphi\), so there exists \(u\in L(r(w))\)
such that \(u \models \varphi\). Since \(w\) is finite, \(L(r(w))=\{w\}\), and
thus \(w \models \varphi\).

If \(b=-\), then \((r(w),-)\) fits \(\varphi\), so no word in \(L(r(w))\)
satisfies \(\varphi\). By Lemma~\ref{lem:omega-to-star}, this implies
\(w \not\models \varphi\).

Thus \(S\) fits \(\varphi\).
\end{proof}

Combining the finite characterization result for
$\LTL_{\ltlFS,\land,\lor,\top,\bot}$ with the schematic translation
yields a compact representation based on schematic examples.

\thmFinitecharacterizationSimpleSchematic*

\begin{proof}
It is known that every formula
$\varphi \in \LTLfin_{\ltlFS,\land,\lor,\top,\bot}[\ap]$
admits a finite characterization by a sample $S$ consisting of
positively labeled finite words and negatively labeled $\omega$-block words.

Let $S' := Schematic(S)$ be the corresponding schematic sample.
By \cref{lem:schematic-correctness}, $S'$ fits $\varphi$,
and for every formula $\psi$,
$S'$ fits $\psi$ if and only if $S$ fits $\psi$.
Hence $S'$ is a finite characterization of $\varphi$.

Since each schematic example in $S'$ is obtained by replacing
$\omega$-blocks by Kleene-star blocks, all regular expressions
in $S'$ have star-height~$1$.
\end{proof}

\paragraph*{Note.}
\Cref{lem:schematic-correctness} does \emph{not} extend to full LTL.
Consider the formula
\[
\varphi \;:=\; \neg\bigl(\ltlG(p \rightarrow \ltlX p)\bigr).
\]
The $\omega$-word $\{p\}^{\omega}$ is a negative example for $\varphi$, since
$\{p\}^{\omega} \models \ltlG(p \rightarrow \ltlX p)$.
However, for every $i \in \mathbb{N}$, the finite word $\{p\}^i$ satisfies $\varphi$,
as the formula $\ltlG(p \rightarrow \ltlX p)$ is violated at the final position.
Consequently,
\[
\{p\}^{\omega} \not\models \varphi
\quad\text{but}\quad
\forall i \in \mathbb{N},\ \{p\}^i \models \varphi.
\]
Thus, replacing $\omega$-block words by their schematic star-unfoldings
does not preserve satisfaction in full LTL, showing that
\cref{lem:schematic-correctness} crucially relies on the fragment
$\LTL_{\ltlFS,\land,\lor,\top,\bot}[\ap]$.

\section{Proofs for Section~\ref{sec:characteristic-samples}}

\thmUniqueImpliesCharacteristic*

\begin{proof}
    Given a formula $\psi$, the teacher outputs the set of examples uniquely characterizing $\psi$. Given a set $E$ the learner looks for the smallest formula $\psi$ that fits $E$ (assuming some arbitrary enumeration of all formulas) and outputs it.
\end{proof}

\thmLTLcharsamples*
\begin{proof}
The teacher and learner agree on some enumeration of all LTL formulas (say in shortlex order).\footnote{The shortlex order compares strings first by length and then lexicographically: shorter strings precede longer ones, and among strings of equal length, the lexicographic order is used.}
The teacher, given a formula $\psi_j$, constructs the characteristic sample for it  by comparing it to all formulas $\psi_i$ for $i<j$ and adding to the sample an example $x\in\sema{\psi_i}\oplus \sema{\psi_j}$ with label $+$ iff $x\in\sema{\psi_j}$. The learner, given  a sample $E$ looks for the smallest $i$ for which $\psi_i$ fits $E$ and returns it. Note that if $E_j$ is the characteristic sample for $\psi_j$ and $E'\supseteq E_j$ extends $E_j$ with additional examples that are consistent with $\psi_j$, the learner on input $E'$ will still output $\psi_j$ since no formula $\psi_i$ for $i<j$ fits $E'$ and all examples in $E'$ are consistent with $\psi_j$.
\end{proof}

\propDFAlowerbound*
\begin{proof}[Proof sketch]
    Consider the family of languages where 
    \[L_n = \{ x\sharp w \sharp y \$ w \sharp ^\omega \mid x,y\in\{0,1,\sharp\}^*, w\in\{0,1\}^n \}\]
    This language without the $\sharp^\omega$ suffix was shown to require $2^{2^{n}}$ states~\cite{ChandraS76}. Intuitively since the DFA needs to remember all words $w\in\{0,1\}^n$ that appear before the $\$$ sign. Since the periodic part plays no interesting role in this language the same is true for the DFA for $(L_n)_{\$}$. %

    On the other hand it can be stated
    in LTL by a conjunction of two formulas, one checking there is only one $\$$ and a tail of $\sharp$'s and the other one checking that there exists an index where $\sharp$ holds and the $i$-th letter afterwards is $0$  (resp. $1$) if and if the $i$-th letter after $\$$ is  $0$  (resp. $1$) for $1\leq i \leq n$.
\end{proof}

\end{document}